\documentclass[preprint,12pt]{elsarticle}
\pdfoutput=1
\usepackage[ruled]{algorithm2e}

\SetAlFnt{\small}
\SetAlCapFnt{\small}
\SetAlCapNameFnt{\small}
\SetAlCapHSkip{0pt}
\IncMargin{-\parindent}

\usepackage{amsthm}
\usepackage{amsmath}
\usepackage{amssymb} 
\usepackage{graphicx}

\usepackage{mathrsfs}

\usepackage[T1]{fontenc}
\usepackage[utf8]{inputenc}
\usepackage{esint}
\journal{ar\!Xiv.org for the record as a manuscript.
	}

\begin{document}

\newtheorem{theorem}{Theorem}[section]
\newtheorem{thm}[theorem]{Theorem}
\newtheorem{df}[theorem]{Definition}
\newtheorem{lem}[theorem]{Lemma}
\newtheorem{corollary}[theorem]{Corollary}
\newtheorem{cor}[theorem]{Corollary}
\newtheorem{mth}[theorem]{The Main Theorem}
\newtheorem{mapth}[theorem]{The Mapping Theorem}
\newtheorem{mlm}[theorem]{The Main Lemma}
\newtheorem{remlm}[theorem]{The ${\tt RemoveMax} $ Invertibility Lemma}
\newtheorem{racrlm}[theorem]{The ${\tt RemoveAll} $ Cost Lemma}
\newtheorem{rmcrlm}[theorem]{The ${\tt RemoveMax} $ Cost Lemma}
\newtheorem{losslm}[theorem]{Credit Loss Characterization Lemma}
\newtheorem{pqlm}[theorem]{The pq Lemma}
\newtheorem{fixlm}[theorem]{Subheap Repair Lemma}
\newtheorem{unlm}[theorem]{The Uniqueness Lemma}
\newtheorem{complem}[theorem]{Competition Lemma}
\newtheorem{diaglem}[theorem]{Diagram Lemma}
\newtheorem{fundlem}[theorem]{Fundamental Lemma}
\newtheorem{charlem}[theorem]{Worst-case Heap Characterization Lemma}
\newtheorem{sLBlm}[theorem]{$ \sum \lambda $ Lower Bound Lemma}
\newtheorem{dth}[theorem]{Decomposition Theorem}
\newtheorem{sith}[theorem]{Singularity Theorem}
\newtheorem{1LBth}[theorem]{The 1$ ^{\mbox{st}} $ Lower Bound Theorem}
\newtheorem{2LBth}[theorem]{The Lower Bound Theorem}
\newtheorem{UBth}[theorem]{The Upper Bound Theorem}
\newtheorem{1oth}[theorem]{The 1$ ^{st} $ Optimality Theorem}
\newtheorem{2oth}[theorem]{The 2$ ^{nd} $ Optimality Theorem}
\newtheorem{3oth}[theorem]{The 3$ ^{rd} $ Optimality Theorem}
\newtheorem{cropt}[theorem]{Optimality Criterion}
\newtheorem{hyp}[theorem]{Hypothesis}
\newtheorem{example}[theorem]{Example}
\newtheorem{property}[theorem]{Property}
\newtheorem{algMergeSort}[theorem]{Algorithm $ {\tt MergeSort} $}
\newtheorem{exercise}[theorem]{Exercise}

\newtheorem{theorem1}{Theorem}[subsection]
\newtheorem{thm1}[theorem1]{Theorem}
\newtheorem{thm2}{Theorem}[section]
\newtheorem{lemma1}[theorem1]{Lemma}
\newtheorem{lemma2}[thm2]{Lemma}
\newtheorem{claim1}[theorem1]{Claim}
\newtheorem{corollary1}[theorem1]{Corollary}
\newtheorem{proposition1}[theorem1]{Proposition}
\newtheorem{problem1}[theorem1]{Problem}
\newtheorem{example1}{Example}[subsection]
\newtheorem{conjecture1}[theorem1]{Conjecture}
\newtheorem{remark1}[theorem1]{Remark}
\newtheorem{property1}[theorem1]{Property}
\newtheorem{algMergeSort1}[theorem1]{Algorithm $ {\tt MergeSort} $}

\pagestyle{myheadings}
\markboth{M.A.Suchenek}{M. A. Suchenek:  Best-case Analysis of MergeSort {\tt (MS) v2}}

\begin{frontmatter}

\title{Best-case Analysis of MergeSort \\ with an Application to the Sum of Digits Problem
\bigskip \\ {\tt \small A manuscript (MS) v2\tnoteref{label1} 
intended for future journal publication\tnoteref{label2}} }
\tnotetext[label2]{\copyright 2016 Marek A. Suchenek.}
\tnotetext[label1]{This is a longer (new Sections~\ref{sec:proof_formula_level_Zigzag}, \ref{sec:proof_no_closed_form}, and \ref{sec:comment} added) version v2 of the article \cite{suc:bestmerg} deposited at ArXive on July 15, 2016, under the same title. These new Sections contain the detailed analytic proofs of Theorems~\ref{thm:formula_level_Zigzag}, \ref{thm:no_closed_form}, and \ref{thm:mainB} that were not included in the original version v1 of this article. Other than that, the current version v2 is identical with the original version v1.}
\author{MAREK A. SUCHENEK}


\address{California State University Dominguez Hills,
Department of Computer Science, \\
1000 E. Victoria St., Carson, CA 90747, USA,
  \texttt{Suchenek@csudh.edu}}

\begin{abstract}
%

\noindent An exact formula 
 \[ B(n) = \frac{n}{2}(\lfloor \lg n \rfloor + 1) -  \sum _{k=0} ^{\lfloor \lg n \rfloor} 2^k  \mbox{\textit{Zigzag}}\,(\frac{n}{2^{k+1}}), \]
 where
 \[ \mbox{\textit{Zigzag}} (x)  =  \min (x - \lfloor x \rfloor, \lceil x \rceil - x), \]
for the minimum number $ B(n) $ of comparisons of keys performed by $ {\tt MergeSort} $ on an $ n $-element array is derived and analyzed. The said formula is less complex than any other known formula for the same and can be evaluated in $ O(\log ^{c}) $ time, where $ c $ is a constant. It is shown that there is no closed-form formula for the above. Other variants for $ B(n) $ are described as well.

\smallskip

Since the recurrence relation for the minimum number of comparisons of keys for $ {\tt MergeSort} $ is identical with a recurrence relation for the number of 1s in binary expansions of all integers between $ 0 $ and $ n $ (exclusively), the above results extend to the sum of binary digits problem. 
\end{abstract}


\end{frontmatter}


\tableofcontents

 \newpage

\section{Introduction}

\begin{quotation}
\textit{``One Picture is Worth a Thousand Words''}\\
 \hspace*{\fill}  [An advertisement for the \textit{San Antonio Light} (1918)] 
 \end{quotation}

%

Teaching undergraduate \textit{Analysis of Algorithms} has been a rewarding, although a bit taxing, experience. I was often surprised to learn that many basic problems that clearly belong to its core syllabus had been left unanswered or partially answered. Also, it seemed a bit odd to me that many otherwise decent texts offered  unnecessarily imprecise computations\footnote{A notable exception in this category is  \cite{sedfla:ana}.} of several rather fundamental results.

\medskip

In this article, I pursue a seemingly marginal topic, the best-case behavior of a well-known sorting algorithm $ {\tt MergeSort} $, which pursuit, however, yields some interesting findings that could hardly be characterized as \textit{``marginal.''} It turns out that - contrary to what a casual student of this subject might believe - computing the exact formula for the number of comparisons of keys that $ {\tt MergeSort} $ performs on any $ n $-element array in the best case is not a routine exercise and leads to a problem that gained some notoriety for being a hard nut to crack analytically: the \textit{sum of digits} problem. Even more unexpectedly, a relatively straightforward\footnote{Although not quite \textit{closed-form}.} formula for the said number of comparisons yields an improvement of a well-known answer to this instance of the \textit{sum of digits} problem:

\begin{quote}
How many $ 1 $s appear in binary representations of all integers between (but not including) $ 0 $ and $ n $?
\end{quote} 

\section{$ {\tt MergeSort} $ and its best-case behavior} \label{sec:mergsor}

A call to $ {\tt MergeSort} $ inherits an $ n $-element array $ {\tt A} $ of integers and sorts it non-decreasingly, following the steps described below.

\begin{algMergeSort} \label{alg:mersor}
To 
sort an $ n $-element array $ {\tt A} $ do:

\begin{enumerate}
\item \label{alg:mersor:item1} If $ n \leq 1 $ then return $ {\tt A} $ to the caller,
\item \label{alg:mersor:item2} If $ n \geq 2 $ then
\begin{enumerate}
\item \label{alg:mersor:item2:it1} pass the first $ \lfloor \frac{n}{2} \rfloor $ elements of $ {\tt A} $ to a recursive call to $ {\tt MergeSort} $,
\item pass the last $ \lceil \frac{n}{2} \rceil $ elements of $ {\tt A} $ to another recursive call to $ {\tt MergeSort} $,
\item \label{alg:mersor:item2:it2} linearly merge, by means of a call to $ {\tt Merge} $, the non-decreasingly sorted arrays that were returned from those calls onto one non-decreasingly sorted array  $ {\tt A}^{\prime} $,
\item \label{alg:mersor:item2:it3} return $ {\tt A}^{\prime} $ to the caller.
 \end{enumerate}  
\end{enumerate}
\end{algMergeSort}

%

A Java code of 
$ {\tt Merge} $ is shown on the Figure~\ref{fig:Merge}.

\begin{figure}[h]
\centering
\includegraphics[width=0.7\linewidth]{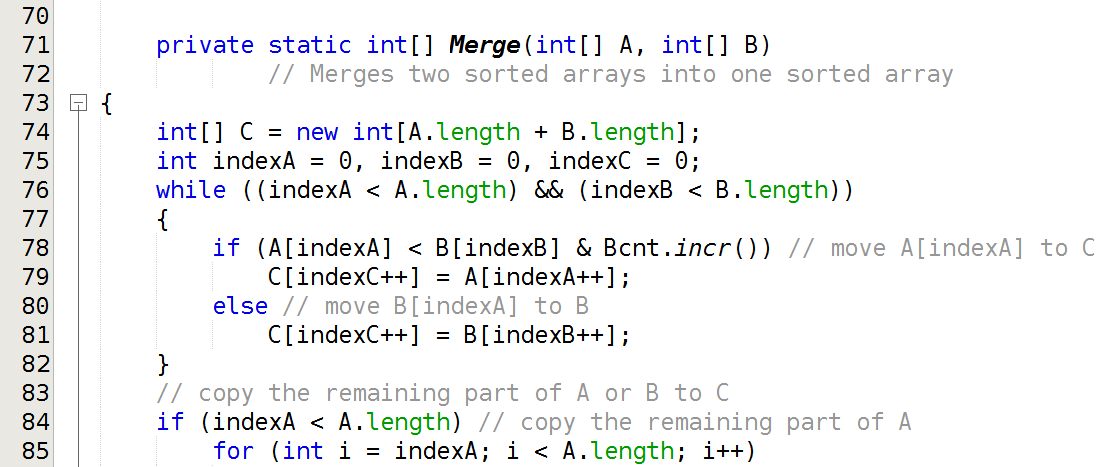} 
\includegraphics[width=0.7\linewidth]{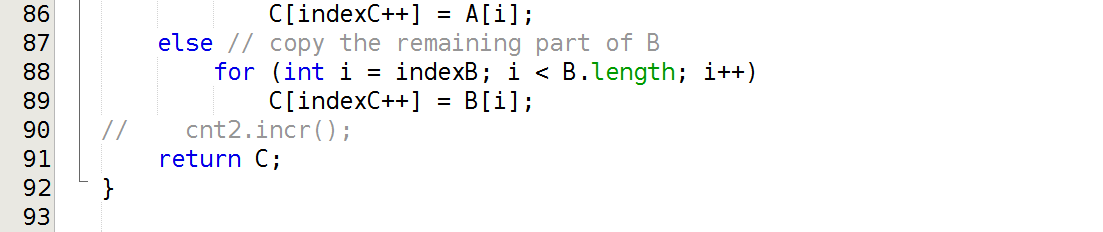}
\caption{A Java code of $ {\tt Merge} $, based on a pseudo-code from \cite{baa:ana}. Calls to $ {\tt Boolean} $ method $ {\tt Bcnt.incr()} $ count the number of comps.}
\label{fig:Merge}
\end{figure}

\medskip

A typical measure of the running time of $ {\tt MergeSort} $ is the number of \textit{comparisons of keys}, which for brevity I call \textit{comps}, that it performs while sorting array $ {\tt A} $. Since no comps are performed outside $ {\tt Merge} $, the running time of $ {\tt MergeSort} $ can be computed as the sum of numbers of comps performed by all calls to $ {\tt Merge} $ during the execution of  $ {\tt MergeSort} $. Since the minimum number of comps performed by $ {\tt Merge} $ on two list 
is 
equal to the length of the shorter list, and any increasingly sorted array on any size $ N \geq 2 $ produces only best-case scenarios for all subsequent calls to $ {\tt Merge} $, a rudimentary analysis of the recursion tree for $ {\tt MergeSort} $ easily yields the exact formula for the minimum number of comps for the entire $ {\tt MergeSort} $. The problem arises when one tries to reduce the said formula, which naturally involves long summations, to one that can be evaluated in a logarithmic time.  

\subsection{Recursion tree} \label{subsecsec:rectre}

The obvious recursion tree for $ {\tt MergeSort} $ and sufficiently large $ n $ is shown on Figure~\ref{fig:rectre}.
\begin{figure} [h]
\includegraphics[scale=.15]{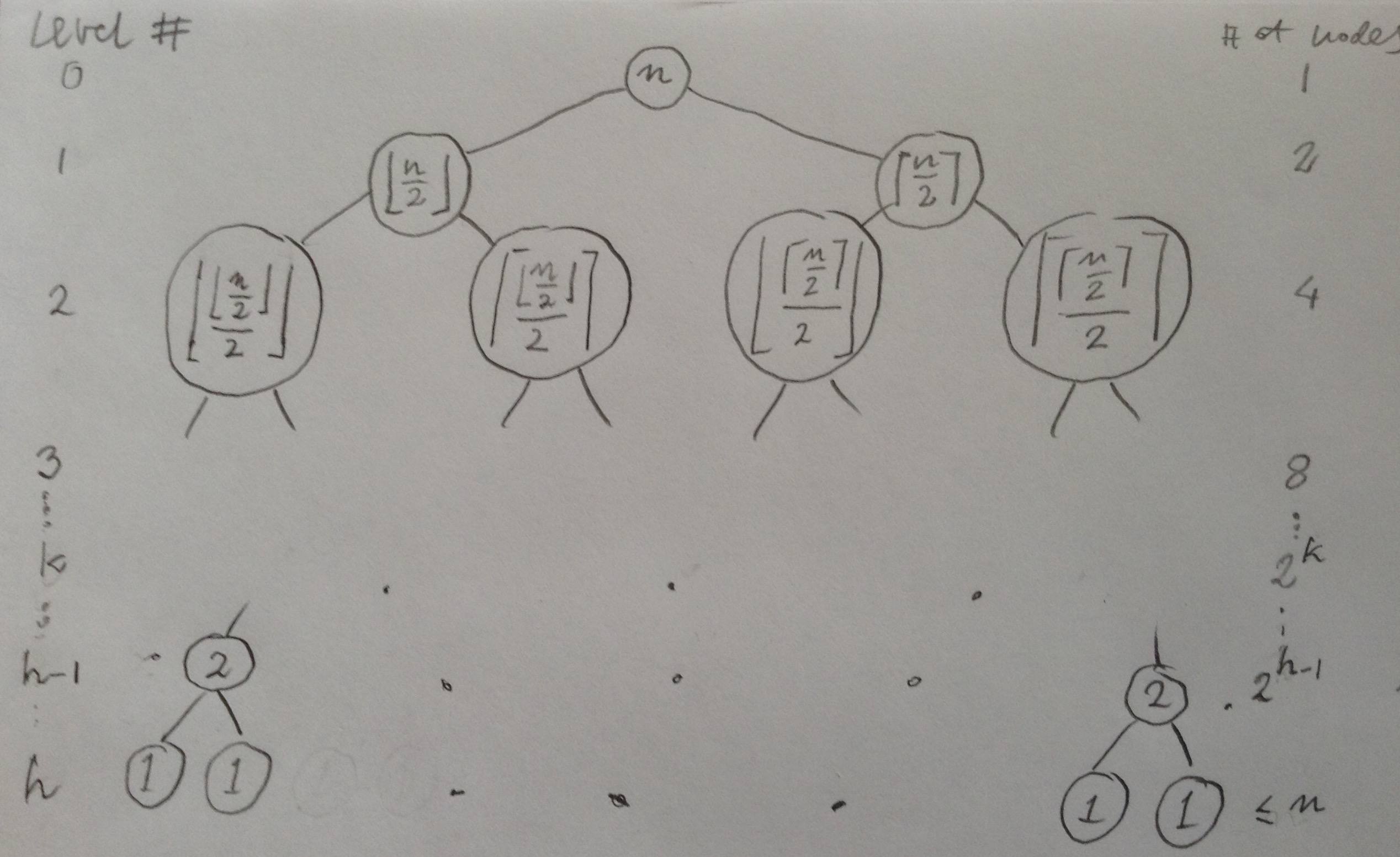} 
\caption{\label{fig:rectre} A sketch of the recursion 2-tree $ T $ for $ {\tt MergeSort} $ for a sufficiently large $ n $, with levels shown on the left and the numbers of nodes shown on the right. The nodes correspond to calls to $ {\tt MergeSort} $ and show sizes of (sub)arrays passed to those calls. The last level is $ h $; it only contains nodes with value 1. The root corresponds to the original call to $ {\tt MergeSort} $. If a call that is represented by a node $ p $ executes further recursive call to $ {\tt MergeSort} $ then these calls are represented by the children of $ p $; otherwise $ p $ is a leaf.}
\end{figure}
A recursive application of the equality\footnote{It can be verified separately for odd and even values of $ n $.}
\begin{equation} \label{eq:n/2=n+1/2}
 \lceil \frac{n}{2}\rceil = \lfloor \frac{n+1}{2}\rfloor 
\end{equation}
allows for rewriting of that tree onto one whose first four levels are shown on Figure~\ref{fig:rectre1}.

\begin{figure} [h]
$ \,\,\,\,\,\,\,\, \,\,\,\,\,\,\,\, \,\,\,\,\,\,\,\, $\includegraphics[width=.75\linewidth]{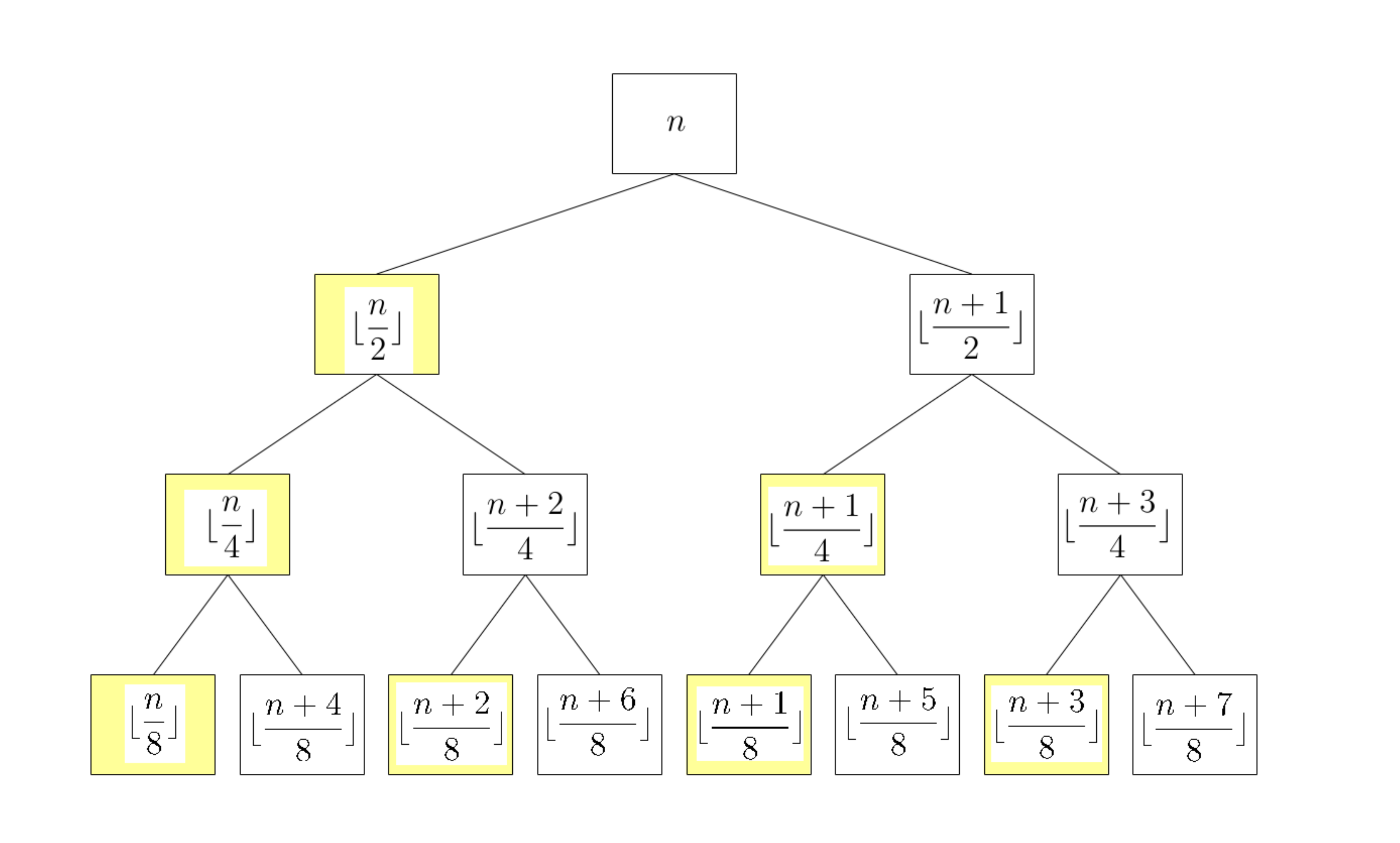}
\caption{\label{fig:rectre1} The first four levels of the recursion 2-tree $ T $ from Figure~\ref{fig:rectre}, with the equality (\ref{eq:n/2=n+1/2}) applied, recursively. The number of comparisons of keys performed in the best case by $ {\tt Merge} $ invoked in step~\ref{alg:mersor:item2:it2} of Algorithm~\ref{alg:mersor} as a result to a call to $ {\tt MergeSort} $ corresponding to a node of $ T $ is equal to the number that is shown in its left child, highlighted yellow. All the right children at level $ k \geq 1 $ of $ T $ show numbers of the form $ \lfloor \frac{n+i}{2^{k+1}} \rfloor $, where $ i \geq 2^k $. Thus all the left children (highlighted yellow) at level $ k \geq 1$ show numbers of the form $ \lfloor \frac{n+i}{2^{k+1}} \rfloor $, where $ i < 2^k $.}
\end{figure}

\subsection{Best-case and its characterization $ B(n) $} \label{subsec:bestchar}

The best-case arrays of sizes $ \lfloor \frac{n}{2} \rfloor $ and $ \lceil \frac{n}{2} \rceil $ for $ {\tt Merge} $, where $ n \geq 2 $, are those in which every element of the first array is less than all elements of the second one. In such a case, $ {\tt MergeSort} $ performs $ \lfloor \frac{n}{2} \rfloor $ of comps.

\medskip


\medskip

Thus the following recurrence relation for the least number $ B(n) $ of comparisons of keys that $ {\tt MergeSort} $ performs on any $ n $-element array is straightforward to derive from its description given by Algorithm~\ref{alg:mersor}.
\begin{equation} \label{eq:rec_rel1}
B(1) = 0,
\end{equation}
and, for $ n \geq 2 $,
\begin{equation} \label{eq:rec_rel2}
B(n) = \lfloor \frac{n}{2}\rfloor + B( \lfloor \frac{n}{2}\rfloor)  + B( \lceil \frac{n}{2}\rceil).
\end{equation}
Using the equality (\ref{eq:n/2=n+1/2}), the recurrence relation (\ref{eq:rec_rel2}) is equivalent to:
\begin{equation} \label{eq:rec_rel3}
B(n) = \lfloor \frac{n}{2}\rfloor + B( \lfloor \frac{n}{2}\rfloor)  + B( \lfloor \frac{n+1}{2}\rfloor)  .
\end{equation}
A graph of $ B(n) $ is shown on Figure~\ref{fig:Best-case_rec_solution}.

\begin{figure} [h]
\centering
\includegraphics[width=0.7\linewidth]{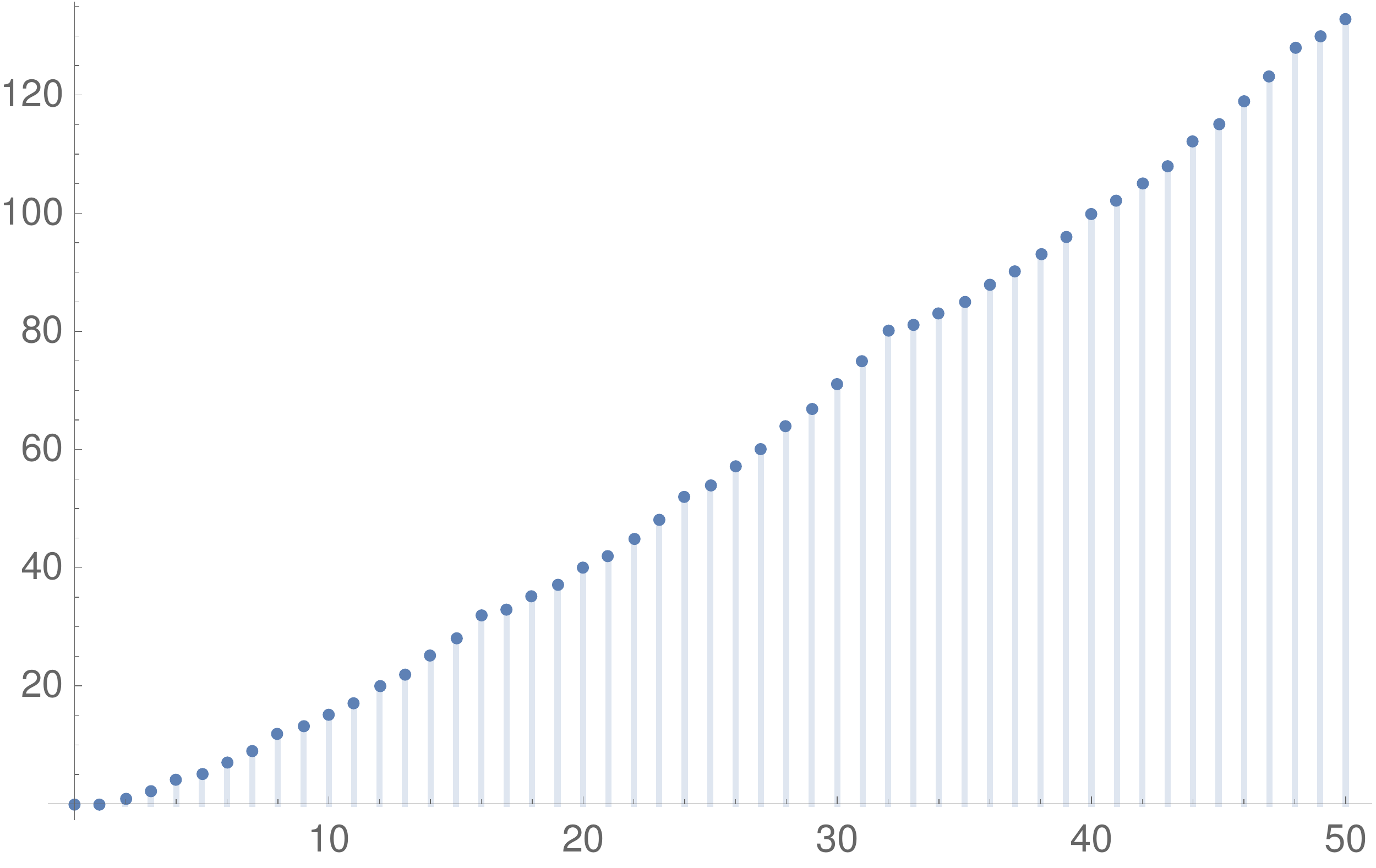}
\caption{Graph of the solution $ B(n) $ of the recurrence \eqref{eq:rec_rel1} and \eqref{eq:rec_rel2}.}
\label{fig:Best-case_rec_solution}
\end{figure}

\medskip

Unfolding the recurrence (\ref{eq:rec_rel3}) allows for noticing that the minimum number $ B(n) $ of comps performed by all calls to $ {\tt Merge} $ is equal to the sum of all values shown at nodes highlighted yellow in the recursion tree $ T $ of Figure~\ref{fig:rectre1}. They may be summed-up level-by-level.
One can notice from Figure~\ref{fig:rectre1} that the number of comps performed at any level $ k $ with the maximal number $ 2^k $ of nodes is given by this formula:
\begin{equation} \label{eq:formula_level_i}
\sum _{i=0} ^{2^k-1} \lfloor \frac{n+i}{2^{k+1}} \rfloor .
\end{equation}

What is not clear is whether all levels of the recursion tree $ T $ are maximal. Fortunately, the answer to this question does not depend on whether given instance of $ {\tt MergeSort} $ is running on a best-case array or on any other case of array. It has been known form a classic analysis of the worst-case running time of $ {\tt MergeSort} $ that every level of its recursion tree $ T $ that contains at least one non-leaf, or - in other words - a node that shows value $ p \geq 2 $, is maximal. \ref{sec:mergesort} page~\pageref{sec:mergesort} contains a detailed derivation of that fact. Thus all levels 0 through $ h-1 $ of $ T $ are maximal.
Therefore, the formula (\ref{eq:formula_level_i}) gives the number of comps for every level $ 0 \leq k \leq h-1 $.

\medskip

 The last level $ h $ of $ T $ may be not maximal because the level $ h-1 $ may contain leaves, or - in other words - nodes that show value $ p = 1 $, where $ p = \lfloor \frac{n+i}{2^{h-1}} \rfloor $ for some $ 0 \leq i \leq 2^{h-1} - 1 $, and as such do not have any children in level $ h $. However, for each such node the value of $ \lfloor \frac{p}{2} \rfloor = \lfloor \frac{n+i}{2^{h}} \rfloor $ is 0, so it can be included in summation (\ref{eq:formula_level_i}) without affecting its value even though the said value does not correspond to any node in level $ h $.
 Therefore, the formula (\ref{eq:formula_level_i}) gives the number of comps for level $  k = h $.

\medskip

   Also, the depth of $ T $ is $ \lfloor \lg n \rfloor $, as the Theorem~\ref{thm:depthRectree} page \pageref{thm:depthRectree} in \ref{sec:excerpts} states.
Thus the minimum number of comps performed by $ {\tt MergeSort} $ is given by this formula:

\begin{equation} \label{eq:formula1}
\sum _{k=0} ^{\lfloor \lg n \rfloor} \sum _{i=0} ^{2^k-1} \lfloor \frac{n+i}{2^{k+1}} \rfloor .
\end{equation}
Unfortunately, the summation~(\ref{eq:formula1}) contains $ n - 1 $ non-zero terms, so it cannot be evaluated quickly in its present form. Fortunately, its inner summation (\ref{eq:formula_level_i}) can be reduced to a closed-form formula. 

\subsection{\textit{Zigzag} function} \label{subsec:zig}

In order to reduce (\ref{eq:formula_level_i}) to a closed form, I am going to use function \textit{Zigzag} defined by:
\begin{equation} \label{eq:def_zigzag}
 \mbox{\textit{Zigzag}}\,(n) = \min (x - \lfloor x \rfloor, \lceil x \rceil - x) .
\end{equation}
 The following fact is instrumental for that purpose.

\medskip

\begin{thm} \label{thm:formula_level_Zigzag}
For every natural number n and every positive natural number m,
\begin{equation} \label{eq:formula_level_Zigzag}
\sum _{i=m} ^{2m-1} \lfloor \frac{n+i}{2m} \rfloor - \sum _{i=0} ^{m-1} \lfloor \frac{n+i}{2m} \rfloor  = 2m \times \mbox{\textit{Zigzag}}\,(\frac{n}{2m}),
\end{equation}
where \textit{Zigzag} is a function defined by \eqref{eq:def_zigzag} and visualized on Figure~\ref{fig:zigzag}.
\end{thm}

\begin{figure} [h]
\centering
\includegraphics[width=0.7\linewidth]{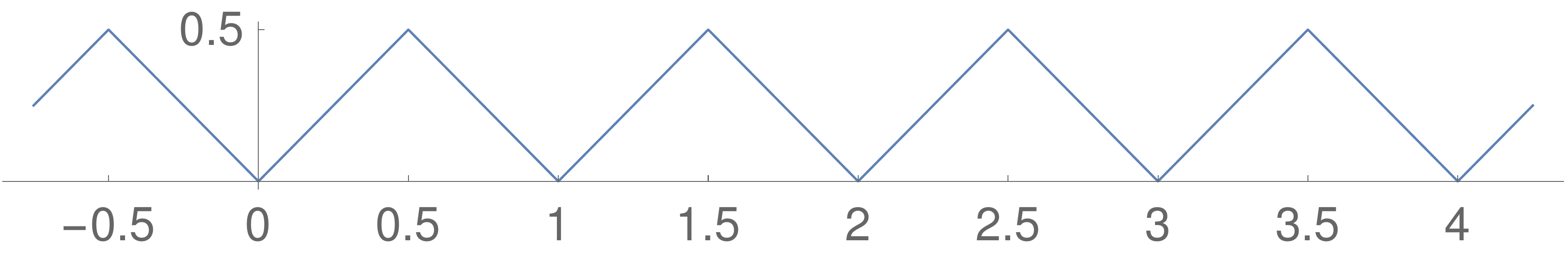}
\caption{Graph of function    $ \mbox{\textit{Zigzag}} (x) $ $ = $  $ \min (x - \lfloor x \rfloor, \lceil x \rceil - x) $.}
\label{fig:zigzag}
\end{figure}

\begin{proof} The equality (\ref{eq:formula_level_Zigzag}) can be verified experimentally, for instance, with a help of software for symbolic computation\footnote{\label{foot:Mat} I used Wolfram Mathematica for that purpose.}. The analytic proof is deferred to Section~\ref{sec:proof_formula_level_Zigzag}.
\end{proof}

\medskip

\begin{corollary} \label{cor:formula_level_Zigzag_1}
For every natural number n and every positive natural number m,
\begin{equation} \label{eq:formula_level_Zigzag_1}
\sum _{i=0} ^{m-1} \lfloor \frac{n+i}{2m} \rfloor  = \frac{n}{2}   -  m \times \mbox{\textit{Zigzag}}\,(\frac{n}{2m}),
\end{equation}
where \textit{Zigzag} is a function defined by \eqref{eq:def_zigzag} and visualized on Figure~\ref{fig:zigzag}.
\end{corollary}
\begin{proof}
First, let's note\footnote{Analytic proof of that fact is a straightforward exercise; see Appendix~\ref{sec:proof1} page \pageref{sec:proof1}.} that
\begin{equation} \label{eq:100}
\sum _{i=0} ^{2m-1} \lfloor \frac{n+i}{2m} \rfloor = n.
\end{equation}
From (\ref{eq:100}) I conclude
\begin{equation} \label{eq:110}
\sum _{i=0} ^{m-1} \lfloor \frac{n+i}{2m}  \rfloor +  \sum _{i=m} ^{2m-1} \lfloor \frac{n+i}{2m}  \rfloor = n.
\end{equation}
Solving equations (\ref{eq:formula_level_Zigzag}) and (\ref{eq:110}) for $ \sum _{i=0} ^{m-1} \lfloor \frac{n+i}{2m}  \rfloor $ yields (\ref{eq:formula_level_Zigzag_1}).
\end{proof}
Here is the closed-form of the summation (\ref{eq:formula_level_i}).
\begin{corollary} \label{cor:formula_level_Zigzag_2}
For every natural number n and every natural number k,
\begin{equation} \label{eq:formula_level_Zigzag_2}
\sum _{i=0} ^{2^k-1} \lfloor \frac{n+i}{2^{k+1}} \rfloor  = \frac{n}{2} -  2^k  \mbox{\textit{Zigzag}}\,(\frac{n}{2^{k+1}}),
\end{equation}
where \textit{Zigzag} is a function defined by \eqref{eq:def_zigzag} and visualized on Figure~\ref{fig:zigzag}.
\end{corollary}
\begin{proof}
Substitute $ m = 2^k $ in (\ref{eq:formula_level_Zigzag_1}).
\end{proof}
The following theorem yields the formula (\ref{eq:formula_Zigzag*}) for the minimum number $ B(n) $ of comps performed by  $ {\tt MergeSort} $.
\begin{thm} \label{thm:formula_Zigzag}
For every natural number n,
\begin{equation} \label{eq:formula_Zigzag*}
\sum _{k=0} ^{\lfloor \lg n \rfloor} \sum _{i=0} ^{2^k-1} \lfloor \frac{n+i}{2^{k+1}} \rfloor  = \frac{n}{2}(\lfloor \lg n \rfloor + 1) -  \sum _{k=0} ^{\lfloor \lg n \rfloor} 2^k  \mbox{\textit{Zigzag}}\,(\frac{n}{2^{k+1}}),
\end{equation}
where \textit{Zigzag} is a function defined by \eqref{eq:def_zigzag} and visualized on Figure~\ref{fig:zigzag}.
\end{thm}
\begin{proof} 
\[\sum _{k=0} ^{\lfloor \lg n \rfloor} \sum _{i=0} ^{2^k-1} \lfloor \frac{n+i}{2^{k+1}} \rfloor  = \sum _{k=0} ^{\lfloor \lg n \rfloor}( \frac{n}{2}  -  2^k  \mbox{\textit{Zigzag}}\,(\frac{n}{2^{k+1}})) = \]
\[ = \frac{n}{2}(\lfloor \lg n \rfloor + 1) -  \sum _{k=0} ^{\lfloor \lg n \rfloor} 2^k  \mbox{\textit{Zigzag}}\,(\frac{n}{2^{k+1}}) .\]
\end{proof}

\medskip

\begin{figure} [h]
\centering
\includegraphics[width=0.7\linewidth]{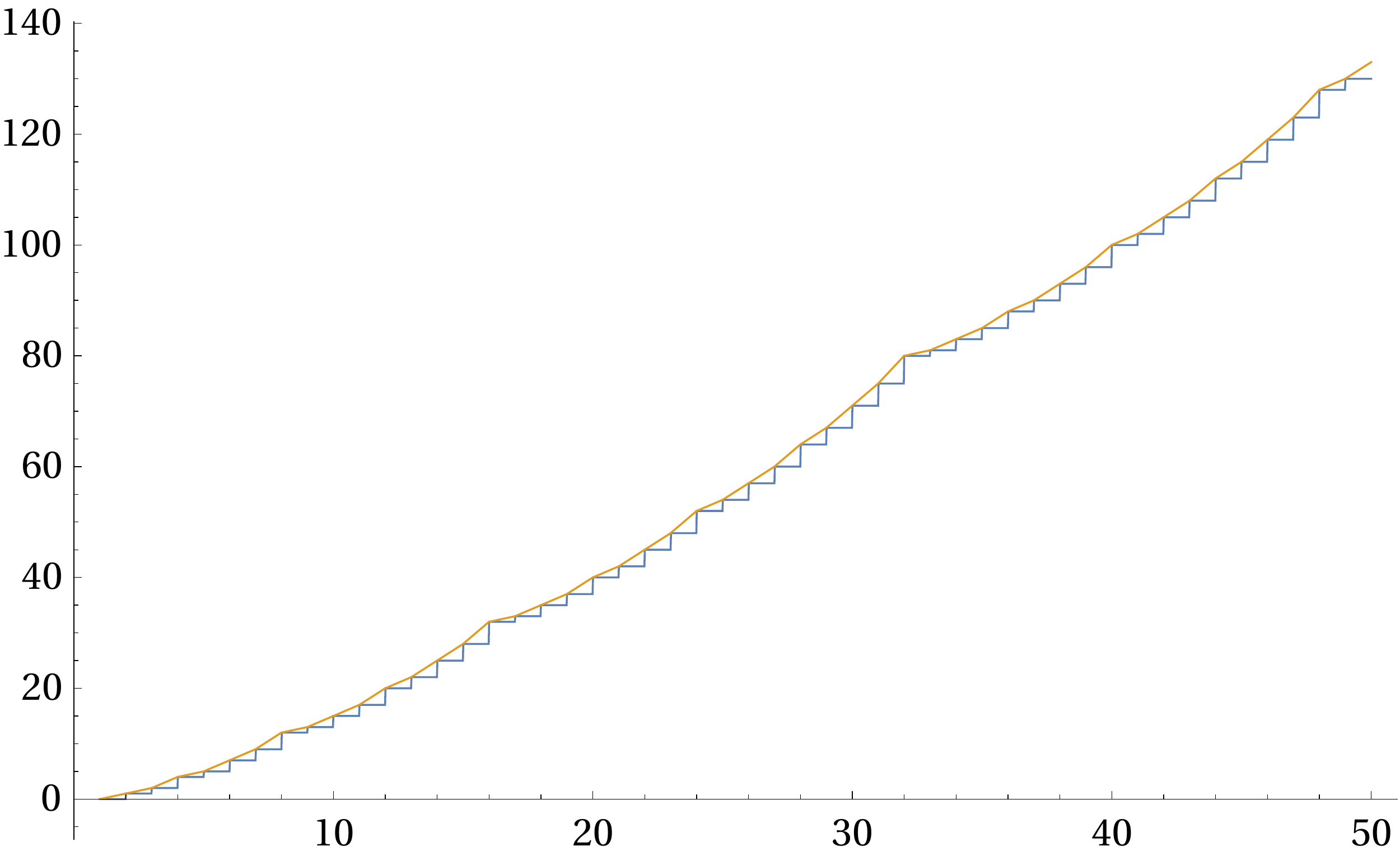}
\caption{Graphs of functions $ \sum _{k=0} ^{\lfloor \lg n \rfloor} \sum _{i=0} ^{2^k-1} \lfloor \frac{n+i}{2^{k+1}} \rfloor  $ (bottom line) and $ \frac{n}{2}(\lfloor \lg n \rfloor + 1) -  \sum _{k=0} ^{\lfloor \lg n \rfloor} 2^k  \mbox{\textit{Zigzag}}\,(\frac{n}{2^{k+1}}) $ (top line) of equality \eqref{eq:formula_Zigzag*}. They coincide with each other for all natural numbers $ n $.}
\label{fig:Best-case_MergeSort_worksheet_best-case}
\end{figure}

Formula (\ref{eq:formula_Zigzag*}), although not quite closed-form, comprises of summation with only $ \lfloor \lg n \rfloor +1 $ closed-form terms, so it may be evaluated in $ O(\log ^{c}) $ time, where $ c $ is a constant. I will show in Section~\ref{sec:frac} that (\ref{eq:formula_Zigzag*}) does not have a closed form.
Graphs of both sides of equality (\ref{eq:formula_Zigzag*}) are shown on Figure~\ref{fig:Best-case_MergeSort_worksheet_best-case}. Once can see that for natural numbers $ n $ they coincide with the solution $ B(n) $ of recurrences \eqref{eq:rec_rel1} and \eqref{eq:rec_rel2} visualized on Figure~\ref{fig:Best-case_rec_solution}.

\begin{cor}
 \label{thm:formula_Zigzag_B}
For every natural number n, the minimum number $ B(n) $ of comps that $ {\tt MergeSort} $ performs while sorting an $ n $-element array is:
\begin{equation} \label{eq:formula_Zigzag_B}
B(n)  = \frac{n}{2}(\lfloor \lg n \rfloor + 1) -  \sum _{k=0} ^{\lfloor \lg n \rfloor} 2^k  \mbox{\textit{Zigzag}}\,(\frac{n}{2^{k+1}}),
\end{equation}
where \textit{Zigzag} is a function defined by \eqref{eq:def_zigzag} and visualized on Figure~\ref{fig:zigzag}.
\end{cor}
%
%
%

\section{A fractal in $ B(n) $} \label{sec:frac}

A deceitfully simple expression

\begin{equation} \label{eq:sum_zigzag}
 \sum _{k=0} ^{\lfloor \lg x \rfloor} 2^{k+1}  \mbox{\textit{Zigzag}}\,(\frac{x}{2^{k+1}}),
\end{equation}
half of which occurs in formula (\ref{eq:formula_Zigzag_B}) of Corollary~\ref{thm:formula_Zigzag_B}, is a formidable adversary for those who may try  to turn it into a closed form, although the time required for its evaluation for any given $ n $ is $ O (\log ^ c) $ \footnote{So, to all practical purposes, (\ref{eq:formula_Zigzag_B}) is a closed-form formula.}. That does not come as a surprise, taking into account that its graph, shown on Figure~\ref{fig:Best-case_MergeSort_worksheet_sum_zigzag}, bears a resemblance of fractal. This can be easily seen as soon as a sawtooth function $ 2^{\lfloor \lg x \rfloor + 1} - x $ is subtracted from it, yielding the function $ F(x) $ given by

\begin{equation} \label{eq:def_F}
F(x) = \sum _{k=0} ^{\lfloor \lg x \rfloor} 2^{k+1}  \mbox{\textit{Zigzag}}\,(\frac{x}{2^{k+1}}) - 2^{\lfloor \lg x \rfloor + 1} + x.
\end{equation}

\medskip

\begin{figure} [h]
\centering
\includegraphics[width=0.7\linewidth]{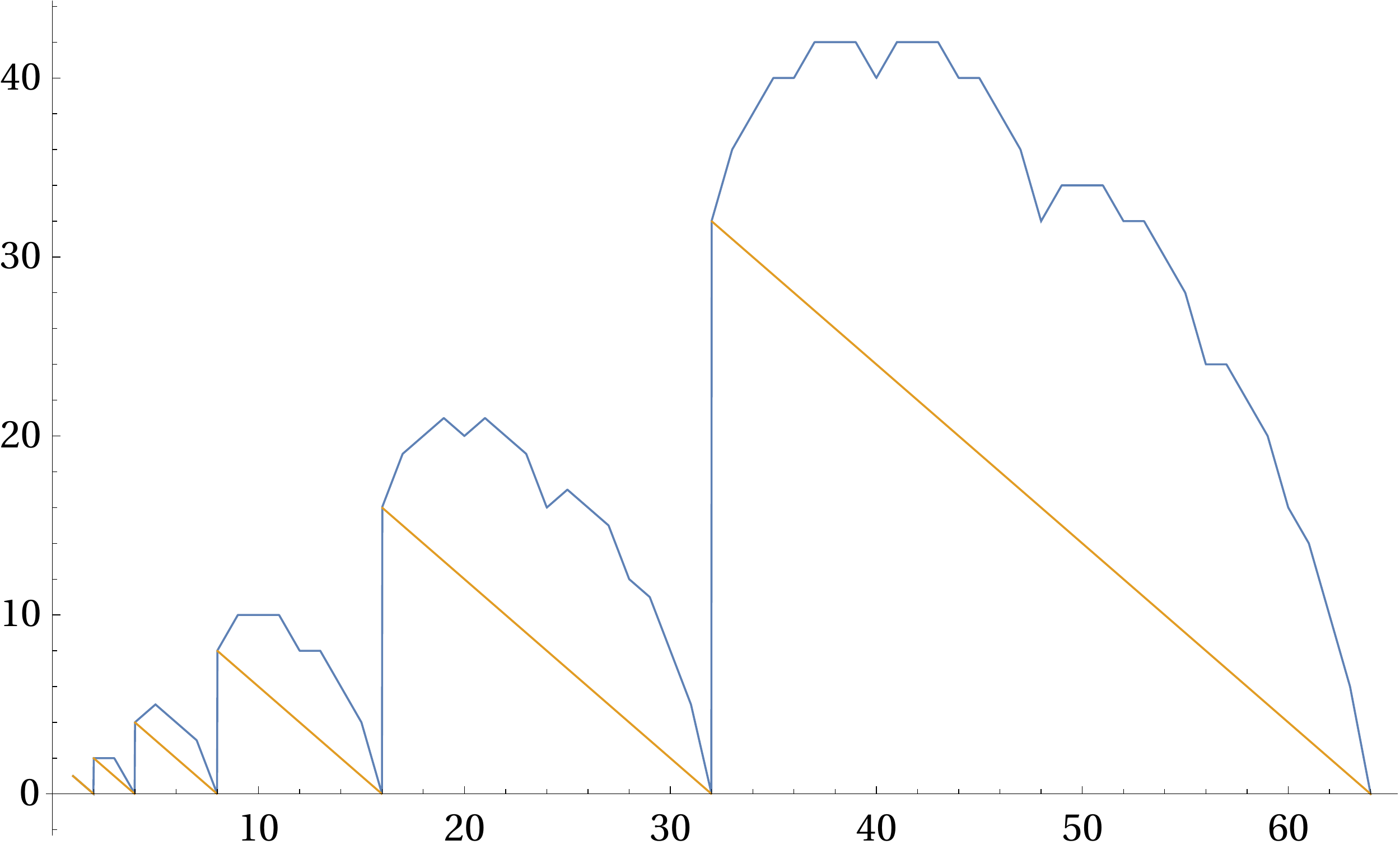}
\caption{A graph of function  $\sum _{k=0} ^{\lfloor \lg x \rfloor} 2^{k+1}  \mbox{\textit{Zigzag}}\,(\frac{x}{2^{k+1}}) $ plotted against a sawtooth function $ 2^{\lfloor \lg x \rfloor + 1} - x $.}
\label{fig:Best-case_MergeSort_worksheet_sum_zigzag}
\end{figure}

\noindent Since $ \frac{1}{2} \leq \frac{x}{2^{\lfloor \lg x \rfloor + 1}} < 1 $, equality \eqref{eq:def_zigzag} implies
\begin{equation} \nonumber
\mbox{\textit{Zigzag}}\,(\frac{x}{2^{\lfloor \lg x \rfloor+1}}) = 1 - \frac{x}{2^{\lfloor \lg x \rfloor+1}},
\end{equation} 
or
\begin{equation} \label{eq:zigprop1}
 2^{\lfloor \lg x \rfloor+1}  \mbox{\textit{Zigzag}}\,(\frac{x}{2^{\lfloor \lg x \rfloor+1}})
= 2^{\lfloor \lg x \rfloor+1} - x.
\end{equation}
The equality \eqref{eq:zigprop1} simplifies definition \eqref{eq:def_F} of function $ F $ to
\begin{equation} \label{eq:fractal}
F(x) = \sum _{k=1} ^{\lfloor \lg x \rfloor} 2^{k}  \mbox{\textit{Zigzag}}\,(\frac{x}{2^{k}}),
\end{equation}
visualized on Figure~\ref{fig:Best-case_MergeSort_worksheet_sum_zigzag_minus_chainsaw}.
\begin{figure} [h]
\centering
\includegraphics[width=0.7\linewidth]{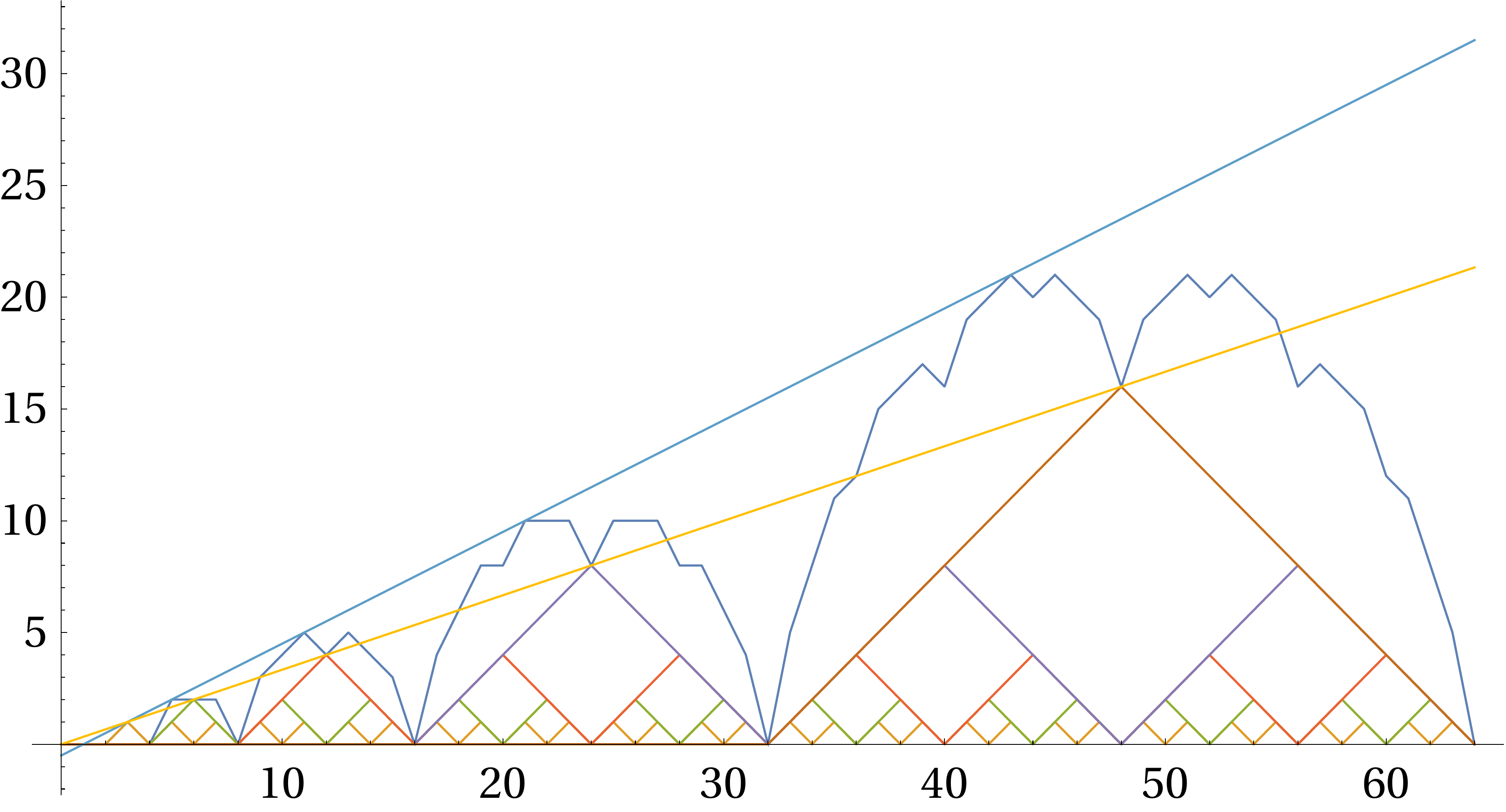}
\caption{A graph of function $  F(x) =  \sum _{k=1} ^{\lfloor \lg x \rfloor} 2^{k}  \mbox{\textit{Zigzag}}\,(\frac{x}{2^{k}})$ plotted below its tight linear upper bound $ y = \frac{x-1}{2} $ (if can be shown that $ F(x) = \frac{x-1}{2} $ whenever  $ x $ $ = $
$ \frac{1}{3}(2^{k+1} + (-1)^k) $ for some integer $ k \geq 0 $); also shown below $  F(x) $ are the terms $ 2^{k} \mbox{\textit{Zigzag}}\,(\frac{x}{2^{k}}) $ of the summation and their tight linear upper bound $ y = \frac{x}{3} $.}
\label{fig:Best-case_MergeSort_worksheet_sum_zigzag_minus_chainsaw}
\end{figure}

\medskip

The function $ F $ is a fractal with quasi similarity that repeats at intervals of exponentially growing length. It is a union 
\begin{equation} \label{eq:F_union_def}
F = \bigcup _{k=0} ^{\infty}  f_k 
\end{equation}
of functions $ f_k $, each having an interval $ [ 2^k, 2^{k+1}) $ as its domain. In other words, for every integer $ k \geq 0 $,
\begin{equation} \label{eq:def_f_k}
f_k = F \restriction [ 2^k, 2^{k+1}),
\end{equation}
which, of course, yields \eqref{eq:F_union_def}.
\begin{figure} [h]
\centering
\includegraphics[width=0.7\linewidth]{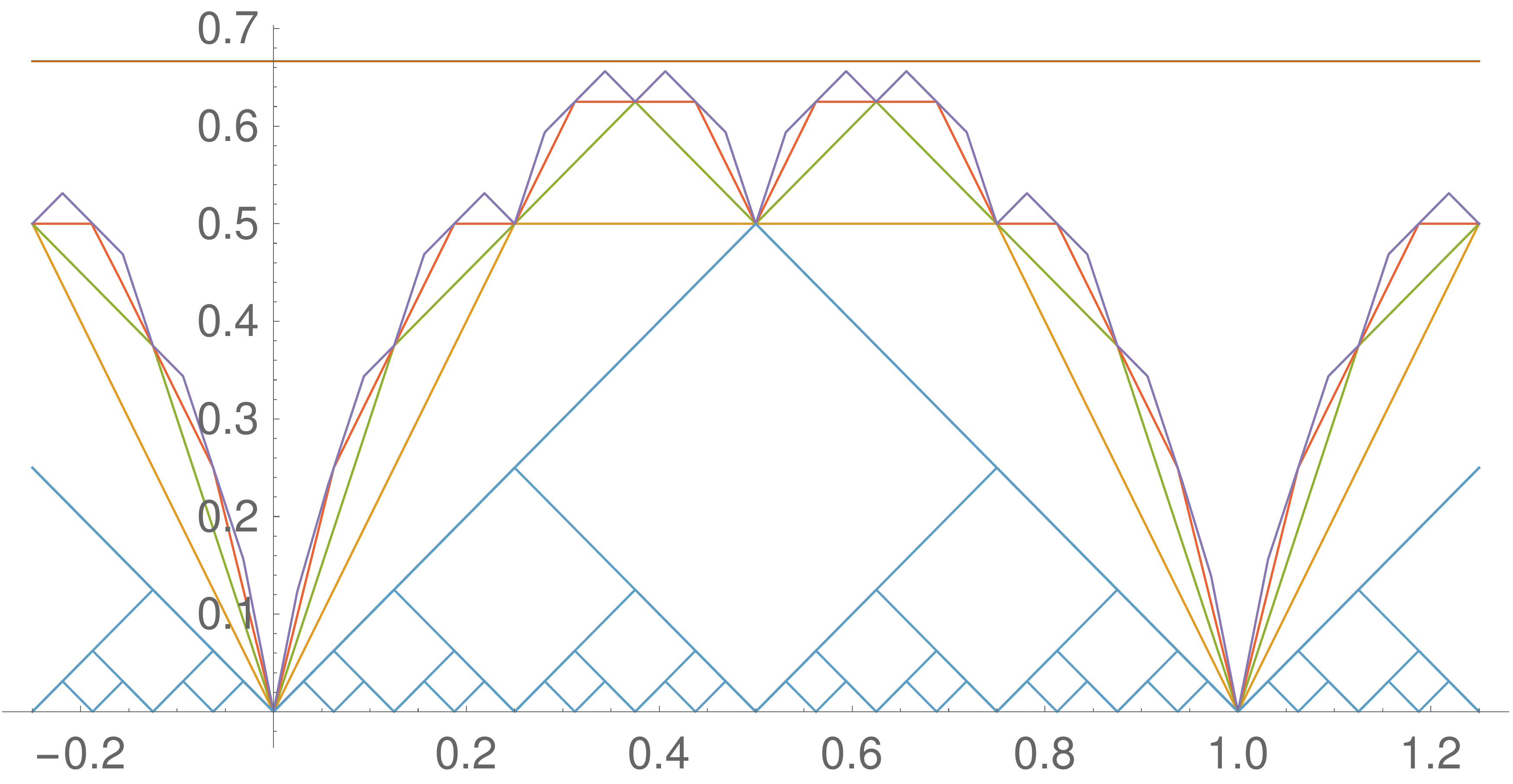}
\caption{A graph of the first six (the first one is $ 0 $) normalized parts of function of Figure~\ref{fig:Best-case_MergeSort_worksheet_sum_zigzag_minus_chainsaw} plotted against the line $ y = \sum _{i=0} ^{\infty} \frac{1}{2^{2i+1}} = \frac{2}{3} $. Also shown (in blue) are the first five terms $ \frac{1}{2^i} \mbox{\textit{Zigzag}} \, (2^i x) $, $ i = 0,...,4 $, of sums that occur in the formula \eqref{eq:f_k_simplified} for $ \tilde{f}_k(x) $; for each integer $ n $ and all $ x \in [n, n+1) $, 
their parts above the $ X $-axis restricted to $ [n, n+1) $ visualize a fragment of an infinite binary search \textit{trie} $ T $ defined as the set of shortest binary expansions of $ x - \lfloor x \rfloor $ with the last digit $ 1 $ (if the said binary expansion is finite) being interpreted as the sequence terminator; in particular, the root of $ T $ is $ .1 $, and if $ {\bf a} $ is a finite binary sequence then the children of binary expansion $ .{\bf a}1 $ are $ .{\bf a}01 $ and $ .{\bf a}11 $ .} 
\label{fig:Best-case_MergeSort_worksheet_sum_zigzag_minus_chainsaw_progression}
\end{figure}

Let $ \hat{f}_k $ be the normalized $ f_k $ on interval $ [0,1) $, defined by:
\begin{equation} \label{eq:f_k_normalized}
\hat{f}_k (x) = \frac{1}{2^k} f_k (2^k (x + 1)) ,
\end{equation} 
and $ \tilde{f}_k $ be the periodized $ \hat{f}_k $ by composing it with a sawtooth function $ x- \lfloor x \rfloor $,\footnote{The fractional part of $ x $.} defined by:
\begin{equation} \label{eq:f_k_periodized}
 \tilde{f}_k(x) = \hat{f}_k(x- \lfloor x \rfloor).
\end{equation}
Contracting definitions \eqref{eq:def_f_k}, \eqref{eq:f_k_normalized}, and \eqref{eq:f_k_periodized}, yields
\begin{equation} \label{eq:f_k_simplified_a}
 \tilde{f}_k(x) = \frac{1}{2^k} F (2^k (x- \lfloor x \rfloor + 1)).
\end{equation}
One can compute\footnote{An elementary geometric argument based on the graph visualized on Figure~\ref{fig:Best-case_MergeSort_worksheet_sum_zigzag_minus_chainsaw_progression} will do.} from \eqref{eq:f_k_simplified_a} the following alternative formula for $  \tilde{f}_k(x) $:
\begin{equation} \label{eq:f_k_simplified}
 \tilde{f}_k(x) = \sum _{i=0} ^{k-1} \frac{1}{2^i} \mbox{\textit{Zigzag}} \, (2^i x).
\end{equation}
Figure~\ref{fig:Best-case_MergeSort_worksheet_sum_zigzag_minus_chainsaw_progression} shows functions $ \tilde{f}_0,...,\tilde{f}_6 $ drawn on the same graph.

\medskip

Since each function $ f_k $, and - therefore - each function $ \hat{f}_k $, and - therefore - each function $  \tilde{f}_k $, are a result of smaller and smaller triangles piled, originating in function \textit{Zigzag} of definition (\ref{eq:fractal}) of function $ F $, on one another as shown on Figure~\ref{fig:Best-case_MergeSort_worksheet_sum_zigzag_minus_chainsaw_progression}, for any integers $ 0 \leq i < j $, $  \tilde{f}_i $ linearly interpolates $  \tilde{f}_j $. Because of that, each $  \tilde{f}_i $ linearly interpolates the limit $ \tilde{F} $ of all $  \tilde{f}_k $s defined by:
\begin{equation} \label{eq:def_lim}
\tilde{F}(x) = \lim _{k \rightarrow \infty}  \tilde{f}_k (x),
\end{equation}
as Figure~\ref{fig:Takagi_function_progression_5} illustrates.
An application of \eqref{eq:f_k_simplified} to \eqref{eq:def_lim} yields:
\begin{equation} \label{eq:def_lim2}
\tilde{F}(x) = \sum _{i=0} ^{\infty} \frac{1}{2^i} \mbox{\textit{Zigzag}} \, (2^i x).
\end{equation}

\begin{figure}[h]
\centering
\includegraphics[width=0.7\linewidth]{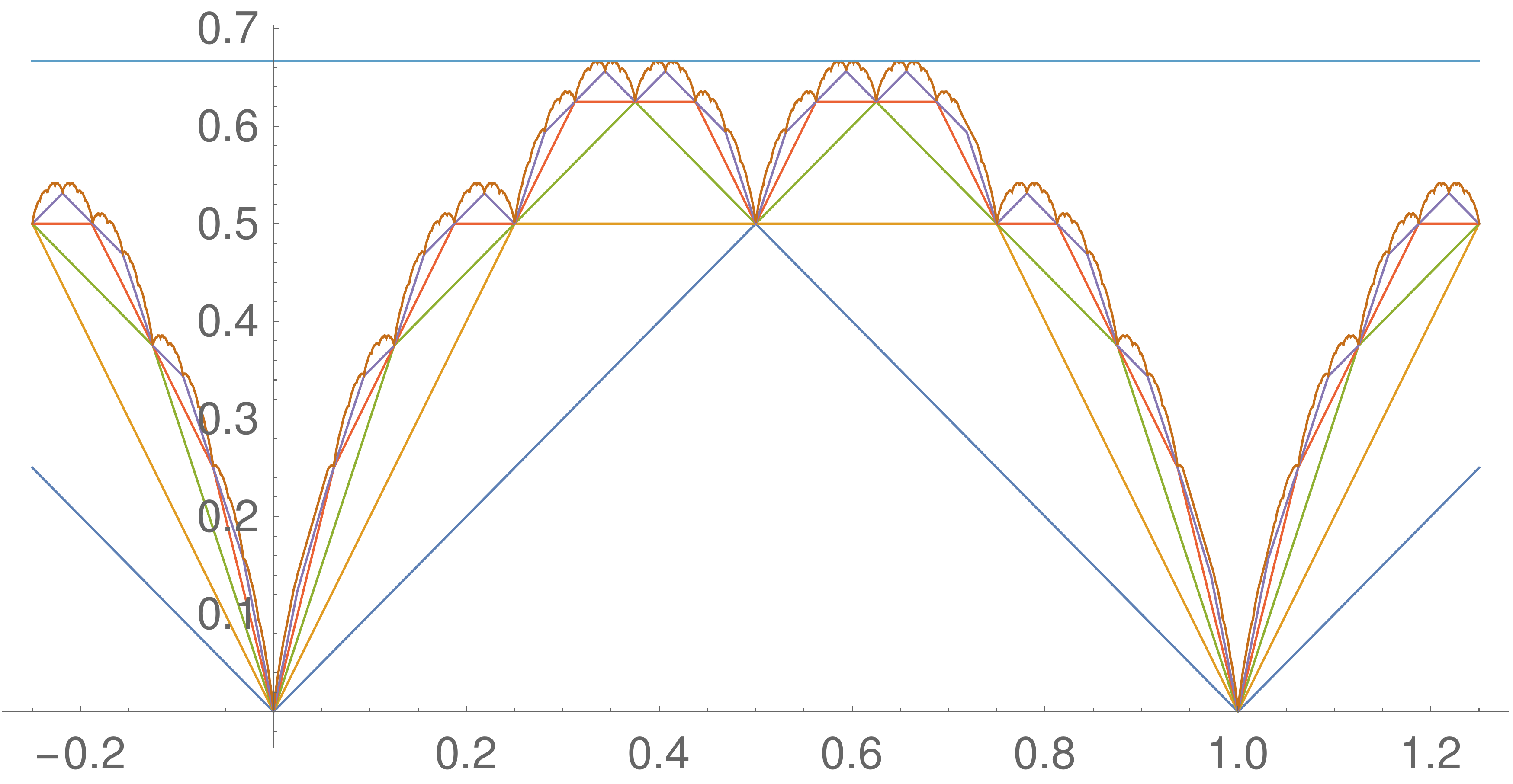}
\caption{Functions functions $ \tilde{f}_0 (x),\tilde{f}_1 (x), ... $ and their limit (the topmost curve) $ \tilde{F} (x) $ given by (\ref{eq:def_lim2}). Collapsing the \textit{Zigzag}$ (x) $ would yield the same, albeit scaled-down (by the factor of 2) pattern $ \frac{1}{2}\tilde{F}(2x) $, as (\ref{eq:def_lim2}) does imply.}
\label{fig:Takagi_function_progression_5}
\end{figure}


Since for every integer $ n $ and $ i \geq k $, $ 2^i \frac{n}{2^k} $ is integer, $ \mbox{\textit{Zigzag}} \, (2^i \frac{n}{2^k}) = 0 $. Therefore, by virtue of \eqref{eq:f_k_simplified} and \eqref{eq:def_lim2}, for every non-negative integer $ k $ and  $ n $,
\begin{equation} \label{eq:Ftilde=ftilde_k}
\tilde{F}(\frac{n}{2^k}) = \tilde{f}_k(\frac{n}{2^k}).
\end{equation}
This and \eqref{eq:f_k_simplified} eliminate the need for infinite summation\footnote{As it appears in \eqref{eq:def_lim2}} while computing $ \tilde{F}(\frac{n}{2^k}) $.

\medskip

It can be shown that although a continuous function, $ \tilde{F} $ is nowhere-differentiable. As such, it does not have a closed-form formula as any closed-form formula  on a real interval must define a function have a derivative at every point of that interval, except for a non-dense set of its points. Since  $ \tilde{F} $ can be expressed in function, described by a closed-form formula, of the right-hand side of formula \eqref{eq:formula_Zigzag*}, the latter does not have a closed-form formula, either.

\begin{thm} \label{thm:no_closed_form}
There is no closed-form formula $ \varphi(n) $ the values of which coincide with $   \sum _{k=0} ^{\lfloor \lg n \rfloor} 2^k  \mbox{\textit{Zigzag}}\,(\frac{n}{2^{k+1}}), $ for all positive integers $ n $, that is, for every closed-form formula $ \varphi(n) $ on function \textit{Zigzag} there is a positive $ n $ such that
\begin{equation} \label{eq:no_closed_form}
  \sum _{k=0} ^{\lfloor \lg n \rfloor} 2^k  \mbox{\textit{Zigzag}}\,(\frac{n}{2^{k+1}}) \neq \varphi(n),
\end{equation}
where \textit{Zigzag} is a function defined by \eqref{eq:def_zigzag} and visualized on Figure~\ref{fig:zigzag}.
\end{thm}
\begin{proof}
Follows from the above discussion. A more detailed proof is deferred to Section~\ref{sec:proof_no_closed_form}. 
\end{proof}

\medskip

This way I arrived at the following conclusion.

\medskip

\begin{cor} \label{cor:no_closed_form}
There is no closed-form formula for $ B(n) $.
\end{cor}
\begin{proof}
A closed-form formula for $ B(n) $ would, by virtue of \eqref{eq:formula_Zigzag_B} page~\pageref{eq:formula_Zigzag_B}, yield a closed-form formula for $ \sum _{k=0} ^{\lfloor \lg n \rfloor} 2^k  \mbox{\textit{Zigzag}}\,(\frac{n}{2^{k+1}}) $, which  by Theorem~\ref{thm:no_closed_form}  does not exist.
\end{proof}
\medskip

\textbf{\textit{Note}}. One can apply the reverse transformations to those used in Section~\ref{sec:frac} on function $ \tilde{F} $  and construct a fractal function $ \breve{F} $, shown on Figure~\ref{fig:Best-case_MergeSort_worksheet_sum_zig_minus_chainsaw_w_upper}, 
given by the equation
\begin{figure}[h]
\centering
\includegraphics[width=0.7\linewidth]{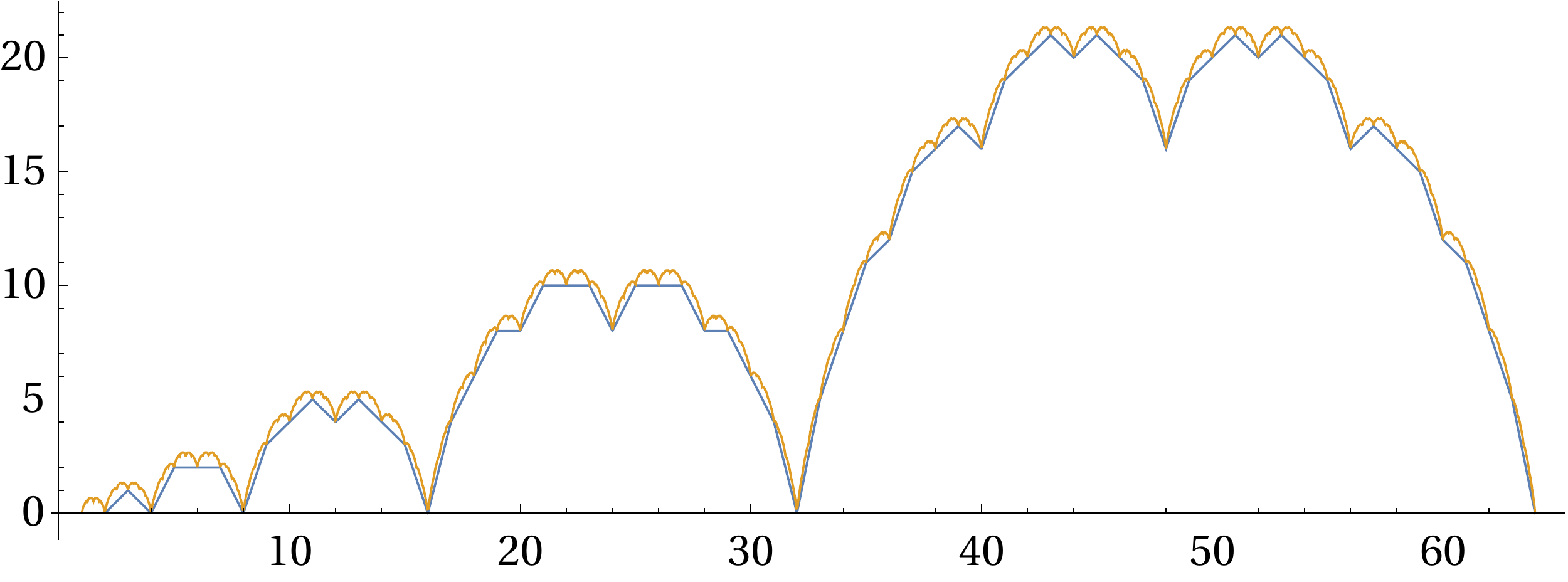}
\caption{A graph of function $ \breve{F} (x) = 2^{\lfloor \lg x \rfloor }  \tilde{F}(\frac{x}{2^{\lfloor \lg x \rfloor }}) $ plotted above a graph of the function
$ F(x) = \sum _{k=1} ^{\lfloor \lg x \rfloor} 2^{k}  \mbox{\textit{Zigzag}}\,(\frac{x}{2^{k}}) $. }
\label{fig:Best-case_MergeSort_worksheet_sum_zig_minus_chainsaw_w_upper}
\end{figure}
\begin{equation} \label{eq:def_F_u}
\breve{F} (x) = 2^{\lfloor \lg x \rfloor }  \tilde{F}(\frac{x}{2^{\lfloor \lg x \rfloor}}),
\end{equation}
that for every positive integer $ n $ satisfies
\begin{equation} \label{eq:F_u=F}
\breve{F} (n) = F(n),
\end{equation}
where $ F $ is given by \eqref{eq:fractal}. 

\section{Computing  $ \tilde{F} (x) $ and $ B(n) $ from one another } \label{sec:compfromfrac}

Computing values of function $ \tilde{F} (x) $ does not have to be as complex as (or more complex than) the definition \eqref{eq:def_lim2} implies. Of course, for every integer $ n $,  $ \tilde{F} (n) $ $ = $ $ 0 $. One can apply some elementary arguments based on a structure visualized on Figure~\ref{fig:Takagi_function_progression_5} to conclude that
\begin{equation}
\tilde{F} (\frac{2}{3}) = \tilde{F} (\frac{1}{3}) = \frac{2}{3} ,
\end{equation}
(the latter being the maximum of $ \tilde{F} (x) $) or that for every positive integer $ k $,
\begin{equation}
\tilde{F} (\frac{1}{2^k}) = \frac{ k }{2^k} .
\end{equation}
It takes a bit more work to compute
\begin{equation}
\tilde{F} (\frac{3}{2^k}) = \frac{3 k - 4 }{2^k} .
\end{equation} 

\medskip

It turns out that computing values of function $ \tilde{F} (x) $ for every $ x $ that has a finite binary expansion can be done easily if an oracle for computing the values of the function $ B(n) $ defined by \eqref{eq:rec_rel1} and \eqref{eq:rec_rel2} is given\footnote{Which is not that surprising after a glance at Figure~\ref{fig:Best-case_MergeSort_worksheet_sum_zig_minus_chainsaw_w_upper}.}. Once that is accomplished, since $ \tilde{F} (x) $ is a continuous function and the set of numbers with finite binary expansions is dense in the set $ \mathfrak{R} $ of reals, it allows for fast approximations of  $ \tilde{F} (x) $ for every real $ x $. \footnote{It helps to remember that $ \tilde{F} $ is a periodic function with $ \tilde{F} (x) = \tilde{F} (x - \lfloor x \rfloor) $.}

\medskip

\begin{thm} \label{thm:mainB}
For every positive integer $ n $ \footnote{Of course, one if free to assume that $ n $ is odd here.} and integer $ k $ with $ n \leq 2^k $,
\begin{equation} \label{eq:mainB}
\tilde{F} (\frac{n}{2^k}) = \frac{n \times k - 2 B(n)}{2^k} .
\end{equation} 
\end{thm}
\begin{proof} The equality \eqref{eq:mainB} can be verified experimentally, for instance, with a help of software for symbolic computation$ ^{\ref{foot:Mat}} $. The analytic proof is deferred to Section~\ref{sec:comment}.
\end{proof}

Theorem~\ref{thm:mainB} allows for easy computing of $ B(n) $ if $ \tilde{F} (\frac{n}{2^k}) $ is given for some $ k \geq \lg n $ using this form of \eqref{eq:mainB}:

\begin{cor} \label{cor:mainB2}
For every positive integer $ n $ and integer $ k $ with $ n \leq 2^k $,
\begin{equation} \label{eq:mainB2}
B(n) =  \frac{n \times k }{2} - 2^{k-1} \tilde{F} (\frac{n}{2^k}).
\end{equation}
\end{cor}
\begin{proof} An obvious conclusion from \eqref{eq:mainB}.
\end{proof}

For instance, putting $ k = \lfloor \lg n \rfloor + 1 $ in \eqref{eq:mainB2} easily yields \eqref{eq:formula_Zigzag_B}. For $ k = \lceil \lg n \rceil $ we obtain
\[  
B(n) =  \frac{n \lceil \lg n \rceil }{2} - 2^{\lceil \lg n \rceil-1} \tilde{F} (\frac{n}{2^{\lceil \lg n \rceil}}) =
  \]
  [by \eqref{eq:def_lim2}]
\[ = \frac{n \lceil \lg n \rceil }{2} - 2^{\lceil \lg n \rceil-1} \sum _{i=0} ^{\infty} \frac{1}{2^i} \mbox{\textit{Zigzag}} \, (2^i \frac{n}{2^{\lceil \lg n \rceil}})  = \]
[since for $ i \geq \lceil \lg n \rceil $, $ 2^i \frac{n}{2^{\lceil \lg n \rceil}} $ is integer and $ \mbox{\textit{Zigzag}} \, (2^i \frac{n}{2^{\lceil \lg n \rceil}})  = 0 $]
\[ = \frac{n \lceil \lg n \rceil }{2} - 2^{\lceil \lg n \rceil-1} \sum _{i=0} ^{\lceil \lg n \rceil - 1} \frac{1}{2^i} \mbox{\textit{Zigzag}} \, (2^i \frac{n}{2^{\lceil \lg n \rceil}})  = \]
\[ =
\frac{n \lceil \lg n \rceil }{2} - \frac{1}{2}  \sum _{i=0} ^{\lceil \lg n \rceil - 1} 2^{\lceil \lg n \rceil - i} \mbox{\textit{Zigzag}} \, ( \frac{n}{2^{\lceil \lg n \rceil - i}})  .\]
Substituting $ k  $ for $ \lceil \lg n \rceil - i $ we conclude
\begin{equation} \label{eq:mainAlt}
B(n) =  \frac{n \lceil \lg n \rceil }{2} - \frac{1}{2}  \sum _{k=1} ^{\lceil \lg n \rceil} 2^{k} \mbox{\textit{Zigzag}} \, ( \frac{n}{2^{k}}),
\end{equation}
a similar to \eqref{eq:formula_Zigzag_B} characterization of $ B(n) $.

%
%
%
%
%
%

\section{Relationship between the best case and the worst case}

A casual student of $ {\tt MergeSort} $ tends to believe that its worst-case behavior is about twice as bad as its best-case behavior. This, of course, is only approximately true. In this Section, I will derive the exact difference between $ 2 B(n) $ and $ W(n) $  using function $ F $ defined by \eqref{eq:def_F} page \pageref{eq:def_F}.

\medskip

An exact formula for the number $ W(n) $  of comparisons of keys performed by $ {\tt MergeSort} $ in the worst case is known\footnote{See \eqref{eq:recmergesort900} in the \ref{sec:excerpts}.} and is given for any positive integer $ n  $ by the following equality:
\begin{equation} \label{eq_worst}
W(n) = \sum _{i = 1} ^{n} \lceil \lg i \rceil .
\end{equation}
From \eqref{eq:formula_Zigzag_B} and \eqref{eq:def_F}, one can derive
\[ 2B(n) =  n  \lfloor \lg n  \rfloor 
- 2^{\lfloor \lg n  \rfloor + 1} + 2n - F(n) = \]
[by $ \sum _{i = 1} ^{n} \lceil \lg i \rceil = n  \lfloor \lg n  \rfloor 
- 2^{\lfloor \lg n  \rfloor + 1} + n + 1 $ from \cite{knu:art}]
\[ = \sum _{i = 1} ^{n} \lceil \lg i \rceil - 1 + n - F(n) =  \]
[by \eqref{eq_worst}]
\[ = W(n) - 1 + n - F(n) . \]
The above yield the following characterization.
\begin{thm} \label{thm:dif2BW}
For every positive integer $ n $, the difference between twice the number $ B(n) $ of comparison of keys performed in the best case and the number $ W(n) $ of comparison of keys performed in the worst case by $ {\tt MergeSort} $ while sorting an $ n $-element array is:
\begin{equation} \label{eq:dif2BW}
2B(n) - W(n) = n - 1 - F(n) ,
\end{equation}
where $ F(n) $, visualized on Figure~\ref{fig:Best-case_MergeSort_worksheet_sum_zigzag_minus_chainsaw}, is given by \eqref{eq:fractal}.
\end{thm}
\begin{proof}
Follows from the above discussion.
\end{proof}
In particular, since for every positive integer $ n $,
\begin{equation} \label{eq:F_ub}
0 \leq F(n) \leq \frac{n-1}{2}
\end{equation}
(see Figure~\ref{fig:Best-case_MergeSort_worksheet_sum_zigzag_minus_chainsaw} for explanation), I conclude with the following tight linear bounds on $ 2B(n) - W(n) $.
\begin{cor}
For every positive integer $ n $, the difference between twice the minimum number $ B(n) $  and the maximum number $ W(n) $ of comparison of keys performed in the worst case by $ {\tt MergeSort} $ while sorting an $ n $-element array satisfies this inequality:
\begin{equation}
\frac{n-1}{2}  \leq 2B(n) - W(n) \leq  n-1 .
\end{equation}
\end{cor}
\begin{proof}
Follows from \eqref{eq:dif2BW} and \eqref{eq:F_ub}.
\end{proof}
\begin{figure}[h]
\centering
\includegraphics[width=0.7\linewidth]{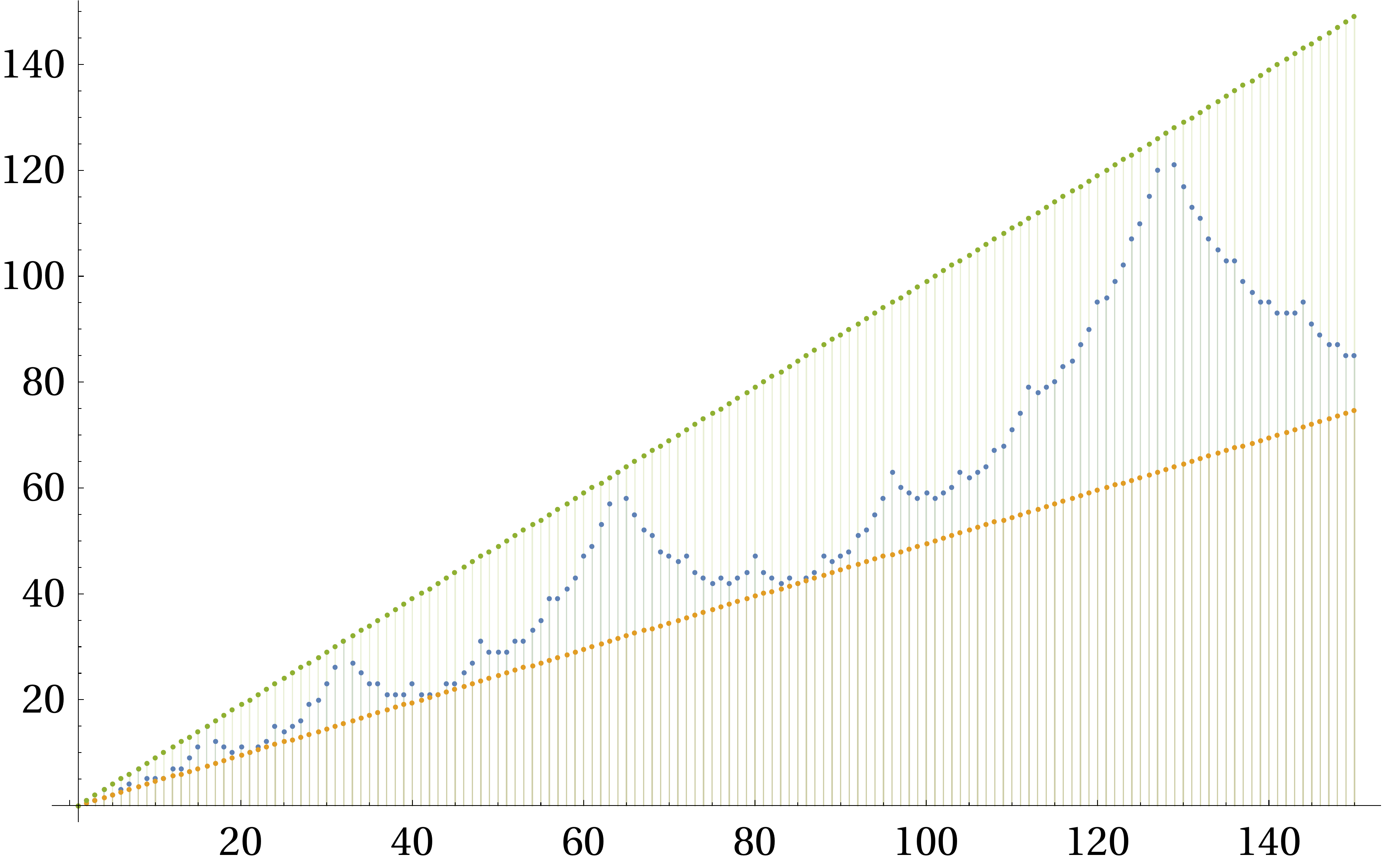}
\caption{A graph of $  2B(n) - W(n)  $ shown between graphs of its tight linear bounds $  n - 1 $ and $ \frac{n - 1}{2} $.}
\label{fig:2B-W_with_bounds}
\end{figure}
Obviously, $  2B(n) - W(n)  = n - 1$ whenever $ F(n) = 0 $, that is, whenever $ n = 2^{\lfloor \lg n \rfloor} $. It can be shown that $  2B(n) - W(n) = \frac{n-1}{2} $ whenever $ n $ $ = $
$ \frac{1}{3}(2^{k+1} + (-1)^k) $ for some integer $ k \geq 0 $.

\medskip

A graph of $ 2B(n) - W(n) $ and its tight bounds are shown on Figure~\ref{fig:2B-W_with_bounds}.

\section{The sum of digits problem} \label{sec:sumdig}

A known explicit formula, published in \cite{tro:digsum}, for the total number 
of bits in all integers between 0 and $ n $ (not including 0 and $ n $) is expressed in terms of function \textit{Zigzag} (referred to as $ 2 g $ in \cite{tro:digsum}) and is given by:\footnote{The following are screen shots and an excerpt from \cite{tro:digsum}.}

{\centering
\includegraphics[width=1\linewidth]{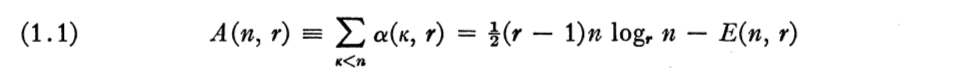}}

\noindent $ \,\, $ Let $ g(x) $ be periodic of period 1 and defined on $ [0,1] $ by

{\centering
\includegraphics[width=1\linewidth]{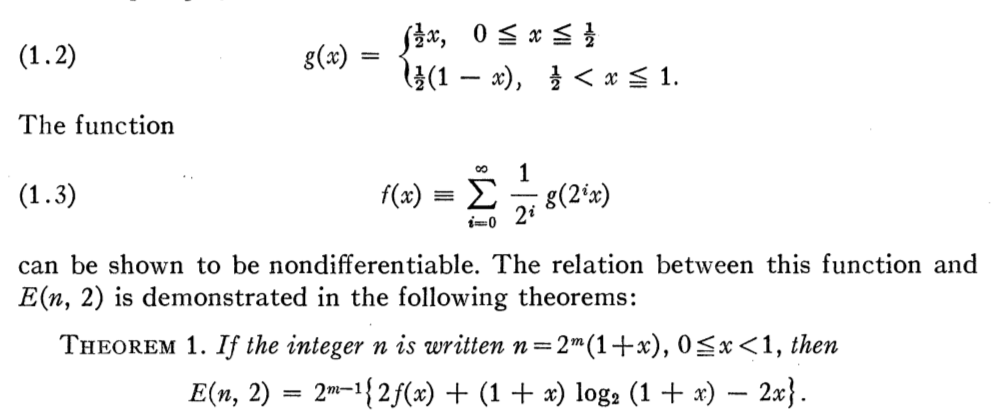}}

\medskip

It has been shown in \cite{mci:numones} that the recurrence relation for $ A(n,2) $  is the same as the recurrence relation for $ B(n) $ given by (\ref{eq:rec_rel1}) and (\ref{eq:rec_rel2}). Therefore, the formula (\ref{eq:formula_Zigzag*}) derived in this paper is equivalent to $ A(n,2) $ given above by the considerably more complicated definition. 
Interestingly, the above definition can be simplified to \eqref{eq:formula_Zigzag*} along the lines of the elementary derivation of the alternative formula \eqref{eq:mainAlt} for $ B(n) $ on page~\pageref{eq:mainAlt}\footnote{Even more interestingly, if someone did bother to simplify Trollope's formula of \cite{tro:digsum} then I am not aware of it.}.

\section{Proof of Theorem~\ref{thm:formula_level_Zigzag} page~\pageref{thm:formula_level_Zigzag}, Subsection~\ref{subsec:zig}} \label{sec:proof_formula_level_Zigzag}

In this Section, I provide an analytic proof of the experimentally-derived Theorem~\ref{thm:formula_level_Zigzag} page~\pageref{thm:formula_level_Zigzag}, Subsection~\ref{subsec:zig} that was instrumental for the derivation of a logarithmic-length formula\footnote{$ B(n) = \frac{n}{2}(\lfloor \lg n \rfloor + 1) -  \sum _{k=0} ^{\lfloor \lg n \rfloor} 2^k  \mbox{\textit{Zigzag}}\,(\frac{n}{2^{k+1}}) $, where $ \mbox{\textit{Zigzag}}\,(x) = \min (x - \lfloor x \rfloor, \lceil x \rceil - x) $.} for $ B(n) $. The result
and its proof have a flavor of Concrete Mathematics. Although they are interesting in their own right, they cannot be found in \cite{knu:concrete}.

%
%
%
%

\begin{thm} \label{thm:formula_level_Zigzag_App} {\rm (Same as Theorem~\ref{thm:formula_level_Zigzag}.)}
For every natural number n and every positive natural number m,
\begin{equation} \label{eq:formula_level_Zigzag_App}
\sum _{i=m} ^{2m-1} \lfloor \frac{n+i}{2m} \rfloor - \sum _{i=0} ^{m-1} \lfloor \frac{n+i}{2m} \rfloor  = 2m \times \mbox{\textit{Zigzag}}\,(\frac{n}{2m}),
\end{equation}
where \textit{Zigzag} is a function defined by \eqref{eq:def_zigzag} and visualized on Figure~\ref{fig:zigzag} page~\pageref{fig:zigzag}.
\end{thm}

\begin{proof} First, let's note that
\[ \sum _{i=m} ^{2m-1} \lfloor \frac{n+i}{2m} \rfloor - \sum _{i=0} ^{m-1} \lfloor \frac{n+i}{2m} \rfloor  =  \sum _{i=0} ^{m-1} \lfloor \frac{n+i+m}{2m} \rfloor - \sum _{i=0} ^{m-1} \lfloor \frac{n+i}{2m} \rfloor  =  \sum _{i=0} ^{m-1} (\lfloor \frac{n+i}{2m} + \frac{1}{2} \rfloor - \lfloor \frac{n+i}{2m} \rfloor)  ,   \]
that is,
\begin{equation} \label{eq:formula_level_Zigzag_App_var}
\sum _{i=m} ^{2m-1} \lfloor \frac{n+i}{2m} \rfloor - \sum _{i=0} ^{m-1} \lfloor \frac{n+i}{2m} \rfloor  =  \sum _{i=0} ^{m-1} (\lfloor \frac{n+i}{2m} + \frac{1}{2} \rfloor - \lfloor \frac{n+i}{2m} \rfloor).
\end{equation}
Let 
\begin{equation} \label{eq:def_r}
 n = k \times 2 m + r , 
\end{equation}where $ 0 \leq r <  2m $, and let $ 0 \leq i <  m $.
We have
\[ \lfloor \frac{n+i}{2m} + \frac{1}{2} \rfloor = \lfloor \frac{k \times 2 m + r+i}{2m} + \frac{1}{2} \rfloor = k + \lfloor \frac{r+i}{2m} + \frac{1}{2} \rfloor \]
and
\[ \lfloor \frac{n+i}{2m}  \rfloor = \lfloor \frac{k \times 2 m + r+i}{2m}  \rfloor = k + \lfloor \frac{r+i}{2m}  \rfloor .\]
Thus, by virtue of \eqref{eq:formula_level_Zigzag_App_var},
\begin{equation} \label{eq:formula_level_Zigzag_App_var2}
\sum _{i=m} ^{2m-1} \lfloor \frac{n+i}{2m} \rfloor - \sum _{i=0} ^{m-1} \lfloor \frac{n+i}{2m} \rfloor 
= \sum _{i=0} ^{m-1} (\lfloor \frac{r+i}{2m} + \frac{1}{2} \rfloor - \lfloor \frac{r+i}{2m} \rfloor) .
\end{equation}
We have
\begin{equation} \label{eq:dif=0_or_1}
\lfloor \frac{r+i}{2m} + \frac{1}{2} \rfloor - \lfloor \frac{r+i}{2m} \rfloor = 
\left\{ \begin{array}{ll}
1 \mbox{ if } \;  \frac{1}{2} \leq \frac{r+i}{2m} < 1  \\ \\
0 \mbox{ otherwise,}
\end{array} \right.
\end{equation}
because  $ \frac{r+i}{2m} + \frac{1}{2} <  \frac{3m}{2m} + \frac{1}{2} = 2 $ so that $ \lfloor \frac{r+i}{2m} + \frac{1}{2} \rfloor \leq 1 $ and, therefore, $ \lfloor \frac{r+i}{2m} + \frac{1}{2} \rfloor - \lfloor \frac{r+i}{2m} \rfloor \leq 1 $.

\medskip
  
Let $ I $ be defined as
\begin{equation} \label{eq:defI}
I = \{ i \in  \mathbb{N} \mid  \frac{1}{2} \leq \frac{r+i}{2m} < 1  \} =
\{ i \in  \mathbb{N} \mid  m-r \leq i < 2m-r  \} .
\end{equation}
By virtue of \eqref{eq:dif=0_or_1}, we have
\begin{equation} \label{eq:formula_level_Zigzag_App_var3}
 \sum _{i=0} ^{m-1} (\lfloor \frac{r+i}{2m} + \frac{1}{2} \rfloor - \lfloor \frac{r+i}{2m} \rfloor) = \sum _{i \in I}  (\lfloor \frac{r+i}{2m} + \frac{1}{2} \rfloor - \lfloor \frac{r+i}{2m} \rfloor)
 = \sum _{i \in I}  1 = \#{I},
\end{equation}
where $ \# (I) $ denotes the cardinality of $ I $.

\medskip

If $ r \leq m $ then, by \eqref{eq:defI}, $ \# I = m - (m-r) = r $. If $ r > m $ then, by \eqref{eq:defI}, $ \# I = 2 m -  r $. In any case,
\begin{equation} \nonumber
\# I = \min (r, 2m - r) = 2m \min (\frac{r}{2m}, 1- \frac{r}{2m}) =
\end{equation}
[since $ 0 \leq \frac{r}{2m} < 1 $ so that $ \lfloor \frac{r+i}{2m} \rfloor = 0 $ and $ \lceil \frac{r+i}{2m} \rceil = 1 $]
\[ =  2m \min (\frac{r}{2m} - \lfloor \frac{r}{2m} \rfloor, \lceil \frac{r}{2m} \rceil- \frac{r}{2m}) =   \]
[by the definition \eqref{eq:def_zigzag} of function  {\em Zigzag}]
\[ = 2m \times \mbox{\em Zigzag} \, ( \frac{r}{2m}) =\]
[since {\em Zigzag} is a periodic function with period 1]
\[=  2m \times \mbox{\em Zigzag} \, ( k+ \frac{r}{2m}) =
  2m \times \mbox{\em Zigzag} \, (  \frac{k \times 2m + r}{2m}) =\]
  [by \eqref{eq:def_r}]
\[=   2m \times \mbox{\em Zigzag} \, (  \frac{n}{2m}) .\] 
Thus
\begin{equation} \label{eq:cardI=2mZ}
\# I = 2m \times \mbox{\em Zigzag} \, (  \frac{n}{2m}) .
\end{equation} 
From \eqref{eq:formula_level_Zigzag_App_var}, \eqref{eq:formula_level_Zigzag_App_var2}, \eqref{eq:formula_level_Zigzag_App_var3}, and \eqref{eq:cardI=2mZ}, I conclude~\eqref{eq:formula_level_Zigzag_App}.
\end{proof}

\section{Proof of Theorem~\ref{thm:no_closed_form} page~\pageref{thm:no_closed_form}, Section~\ref{sec:frac} \label{sec:proof_no_closed_form}}

In this Section, I present a brief discussion/motivation of what can be generally considered a \textit{closed-form formula} for a function from the set of real numbers into a set of real numbers. I provide an analytic proof of Theorem~\ref{thm:no_closed_form} page~\pageref{thm:no_closed_form}, Section~\ref{sec:frac} that implies the non-existence of closed-form formula for the minimum number $ B(n) $ of comparisons of keys by $ {\tt MergeSort} $ while sorting an $ n $-element array. I am going to use the acronym $ cf\!f $ as an abbreviation for \textit{closed-form formula}. 

\medskip

\medskip

For reader's convenience, the Theorem~\ref{thm:no_closed_form} is quoted below as Theorem~\ref{thm:no_closed_form_add}.

\begin{thm} \label{thm:no_closed_form_add} {\rm (Same as Theorem~\ref{thm:no_closed_form}.)}
There is no $ cf\!f $ $ \varphi(n) $ the values of which coincide with 
$   \sum _{k=0} ^{\lfloor \lg n \rfloor} 2^k  \mbox{\textit{Zigzag}}\,(\frac{n}{2^{k+1}}), $ for all positive integers $ n $, that is, for every $ cf\!f $ $ \varphi(n) $ on function \textit{Zigzag} there is a positive $ n $ such that
\begin{equation} \label{eq:no_closed_form_add}
  \sum _{k=0} ^{\lfloor \lg n \rfloor} 2^k  \mbox{\textit{Zigzag}}\,(\frac{n}{2^{k+1}}) \neq \varphi(n),
\end{equation}
where \textit{Zigzag} is a function defined by \eqref{eq:def_zigzag} and visualized on Figure~\ref{fig:zigzag}.
\end{thm}

The rest of this Section constitutes the proof of Theorem~\ref{thm:no_closed_form_add}. 


\medskip

First, let me use an example of function $ 2^x : \mathbb{R} \longrightarrow  \mathbb{R} $ as an insight of what may be accepted as a $cf\!f$ for a \textit{continuous} function - like, say, $ \tilde{F}(x) $ - on the set $ \mathbb{R} $ of reals or on an interval thereof. One picks a dense\footnote{In the metric topology of $ \mathbb{R} $.} subset $ \mathbb{Q} $ of $ \mathbb{R} $,  with  a collection of mappings $ \rho_x(i) : \mathbb{N}  \longrightarrow \mathbb{Q} $, where $ x \in \mathbb{R} $, given by $ \rho_x(i) = \frac{\lfloor i \times x \rfloor}{i} $ so that $ \lim _{i \rightarrow \infty}\rho_x(i) = x $. Since for any $ x \in \mathbb{R} \setminus  \mathbb{Q} $, $ 2^x $ has been defined as 
\begin{equation}
2^x = \lim _{i \rightarrow \infty}  2^{\rho_x (i)} =  \lim _{i \rightarrow \infty}  \sqrt[i]{2^{\lfloor i \times x \rfloor}} ,
\end{equation}
 $ \lim _{i \rightarrow \infty}  \sqrt[i]{2^{\lfloor i \times x \rfloor}} $ is considered a $cf\!f$ $ \alpha : \mathbb{R}  \longrightarrow \mathbb{R}  $ for $ 2^x $.
 
 \medskip
 
%

\medskip

\begin{lem} \label{lem:cff_F_tilde}
For every positive integer $ n $,
\begin{equation}  \label{eq:proof5}
\tilde{F} (\frac{n}{2^{\lfloor \lg n \rfloor}}) = \frac{n ( \lfloor \lg n \rfloor + 2) - 2 B(n)}{2^{\lfloor \lg n \rfloor}}  -2,
\end{equation}
where the function $F$ has been defined by the equality~\eqref{eq:def_F} page~\pageref{eq:def_F}, 
the function $\tilde{F}  $, visualized on Figure~\ref{fig:Tak} \footnote{Also, together with its partial sums, on Figure~\ref{fig:Takagi_function_progression_5}, page~\pageref{fig:Takagi_function_progression_5}.}, has been defined  by the equality~\eqref{eq:def_lim2} page~\pageref{eq:def_lim2}, and the function $ B(n) $ has been defined by the equations \eqref{eq:rec_rel1} and \eqref{eq:rec_rel3} page~\pageref{eq:rec_rel1}.

\end{lem}
\begin{proof}

\begin{figure}
\centering
\includegraphics[width=0.7\linewidth]{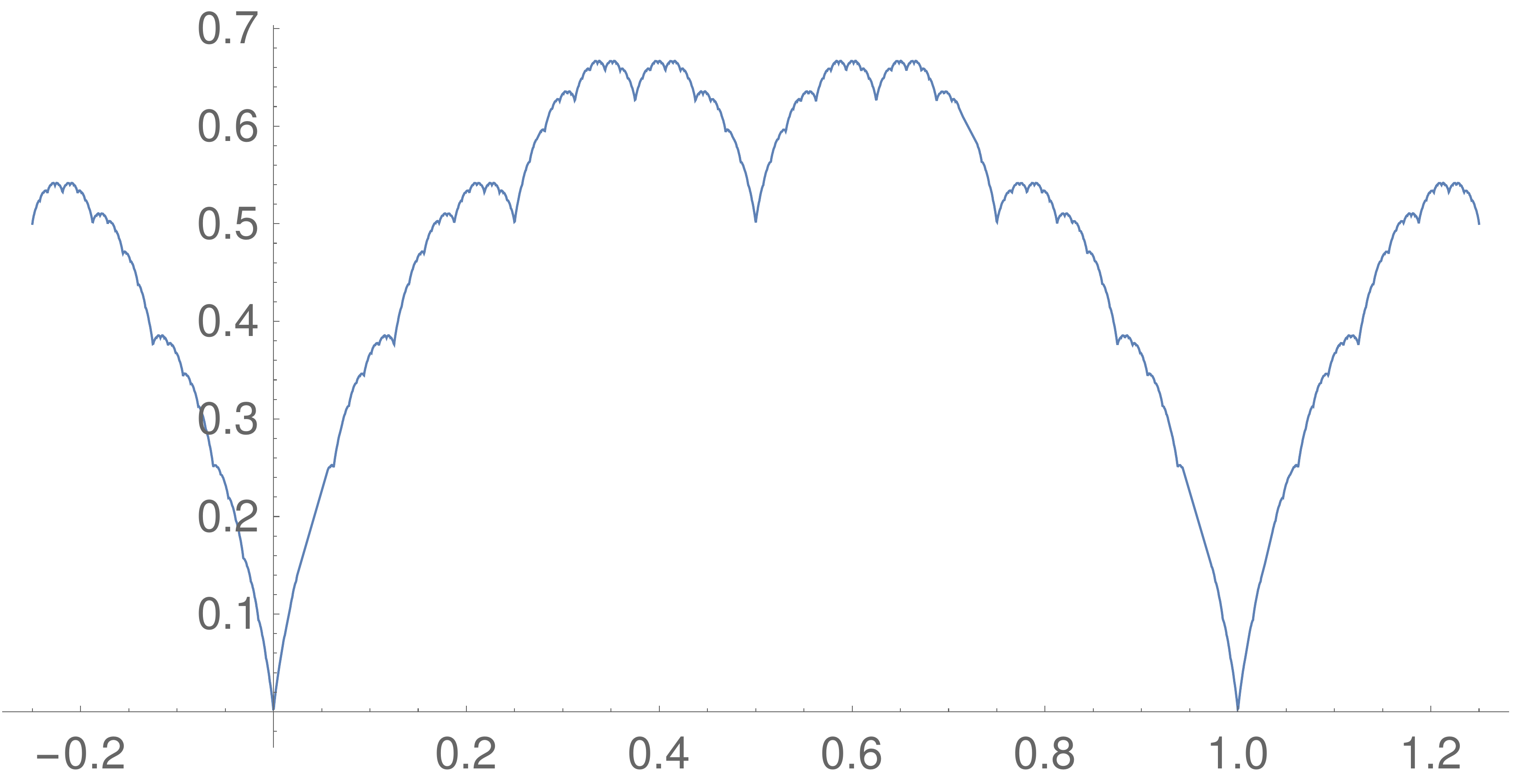} \caption{A graph of the Blancmange function $ \tilde{F}(x) = \sum _{i=0} ^{\infty} \frac{1}{2^i} Zigzag (2^i x) $. \label{fig:Tak}}
\end{figure}
%
%
%
From \eqref{eq:def_F} page~\pageref{eq:def_F}, I compute
\[ \sum _{i=0} ^{\lfloor \lg n \rfloor} 2^{i+1}  \mbox{\textit{Zigzag}}\,(\frac{n}{2^{i+1}}) =
 F(n) + 2^{\lfloor \lg n \rfloor + 1} - n ,\]
that is,
\begin{equation} \label{eq:proof1}
 \sum _{i=0} ^{\lfloor \lg n \rfloor} 2^{i}  \mbox{\textit{Zigzag}}\,(\frac{n}{2^{i+1}}) =
 \frac{1}{2} F(n) + 2^{\lfloor \lg n \rfloor} - \frac{n}{2} .
\end{equation}
Applying \eqref{eq:proof1} to the equality \eqref{eq:formula_Zigzag_B} page~\pageref{eq:formula_Zigzag_B}, I conclude
\[  B(n)  = \frac{n}{2}(\lfloor \lg n \rfloor + 1) - (\frac{1}{2} F(n) + 2^{\lfloor \lg n \rfloor} - \frac{n}{2}) , \]
or
\[  B(n)  = \frac{n}{2}(\lfloor \lg n \rfloor + 1) - \frac{1}{2} F(n) - 2^{\lfloor \lg n \rfloor} + \frac{n}{2} , \]
that is,
\[ \frac{1}{2} F(n)  = \frac{n}{2}(\lfloor \lg n \rfloor + 1) -  B(n) - 2^{\lfloor \lg n \rfloor} + \frac{n}{2} , \]
or
\begin{equation} \label{eq:proof2}
 F(n)  = n (\lfloor \lg n \rfloor + 2) -  2 B(n) - 2^{\lfloor \lg n \rfloor+1}  .
\end{equation}
On the other hand, by virtue of \eqref{eq:fractal} page~\pageref{eq:fractal},
\[  F(n) = \sum _{i=1} ^{\lfloor \lg n \rfloor} 2^{i}  \mbox{\textit{Zigzag}}\,(\frac{n}{2^{i}}) = \]
[putting $ j = \lfloor \lg n \rfloor - i $]
\[ = \sum _{j=0} ^{\lfloor \lg n \rfloor-1} 2^{\lfloor \lg n \rfloor-j}  \mbox{\textit{Zigzag}}\,(\frac{n}{2^{\lfloor \lg n \rfloor-j}}) = 
2^{\lfloor \lg n \rfloor} \sum _{j=0} ^{\lfloor \lg n \rfloor-1} \frac{1}{2^j}  \mbox{\textit{Zigzag}}\,(\frac{2^j n}{2^{\lfloor \lg n \rfloor}}) =\]
[since for $ j \geq \lfloor \lg n \rfloor $, $ \frac{2^j n}{2^{\lfloor \lg n \rfloor}} \in \mathbb{N} $ so that $ \mbox{\textit{Zigzag}}\,(\frac{2^j n}{2^{\lfloor \lg n \rfloor}}) = 0 $]
\[    = 
2^{\lfloor \lg n \rfloor} \sum _{j=0} ^{\infty} \frac{1}{2^j}  \mbox{\textit{Zigzag}}\,(\frac{2^j n}{2^{\lfloor \lg n \rfloor}}) =             \]
[by \eqref{eq:def_lim2} page~\pageref{eq:def_lim2}]
\[  2^{\lfloor \lg n \rfloor} \tilde{F} (\frac{n}{2^{\lfloor \lg n \rfloor}}) . \]
Thus,
\begin{equation} \label{eq:proof3}
F(n) =  2^{\lfloor \lg n \rfloor} \tilde{F} (\frac{n}{2^{\lfloor \lg n \rfloor}})
\end{equation}
or
\begin{equation} \label{eq:proof4}
\tilde{F} (\frac{n}{2^{\lfloor \lg n \rfloor}}) = \frac{1}{2^{\lfloor \lg n \rfloor}} F(n) .
\end{equation}
Combining equalities \eqref{eq:proof2} and \eqref{eq:proof4} yields
\[  \tilde{F} (\frac{n}{2^{\lfloor \lg n \rfloor}}) = \frac{1}{2^{\lfloor \lg n \rfloor}} ( n (\lfloor \lg n \rfloor + 2) -  2 B(n) - 2^{\lfloor \lg n \rfloor+1}) ,  \]
or \eqref{eq:proof5}.
\end{proof}

\begin{lem} \label{lem:main_no_cff}
If the function $ B(n) $ defined by the equations \eqref{eq:rec_rel1} and \eqref{eq:rec_rel3}  has a $ cf\!f $ $ \beta : \mathbb{N}  \longrightarrow \mathbb{N}  $ then the function  $ \tilde{F}(x)  $ defined by the equation \eqref{eq:def_lim2} page~\pageref{eq:def_lim2} has a $ cf\!f $ $ \varphi : [1,2)  \longrightarrow [0,\frac{2}{3}]  $.
\end{lem}
\begin{proof}
Let
\begin{equation} \label{eq:def_D}
D = \{ \frac{n}{2^{\lfloor \lg n \rfloor}} \mid n \in \mathbb{N} \} 
\end{equation}
be the set of rationals in the interval $ [1,2) $ with finite binary expansions\footnote{It is a trivial exercise to show that every real number with finite binary expansion in the interval $ [1,2) $ is of the form $ \frac{n}{2^{\lfloor \lg n \rfloor}} $ for some $ n \in \mathbb{N} $, and it is obvious that every real number of that form has a finite binary expansion and falls into that interval.}, enumerated by $ \nu (n) : \mathbb{N}  \longrightarrow D $ given by $ \nu (n) = \frac{n}{2^{\lfloor \lg n \rfloor}} $ and visualized on Figure~\ref{fig:mapping}.

\begin{figure} [h]
\centering
\includegraphics[width=0.7\linewidth]{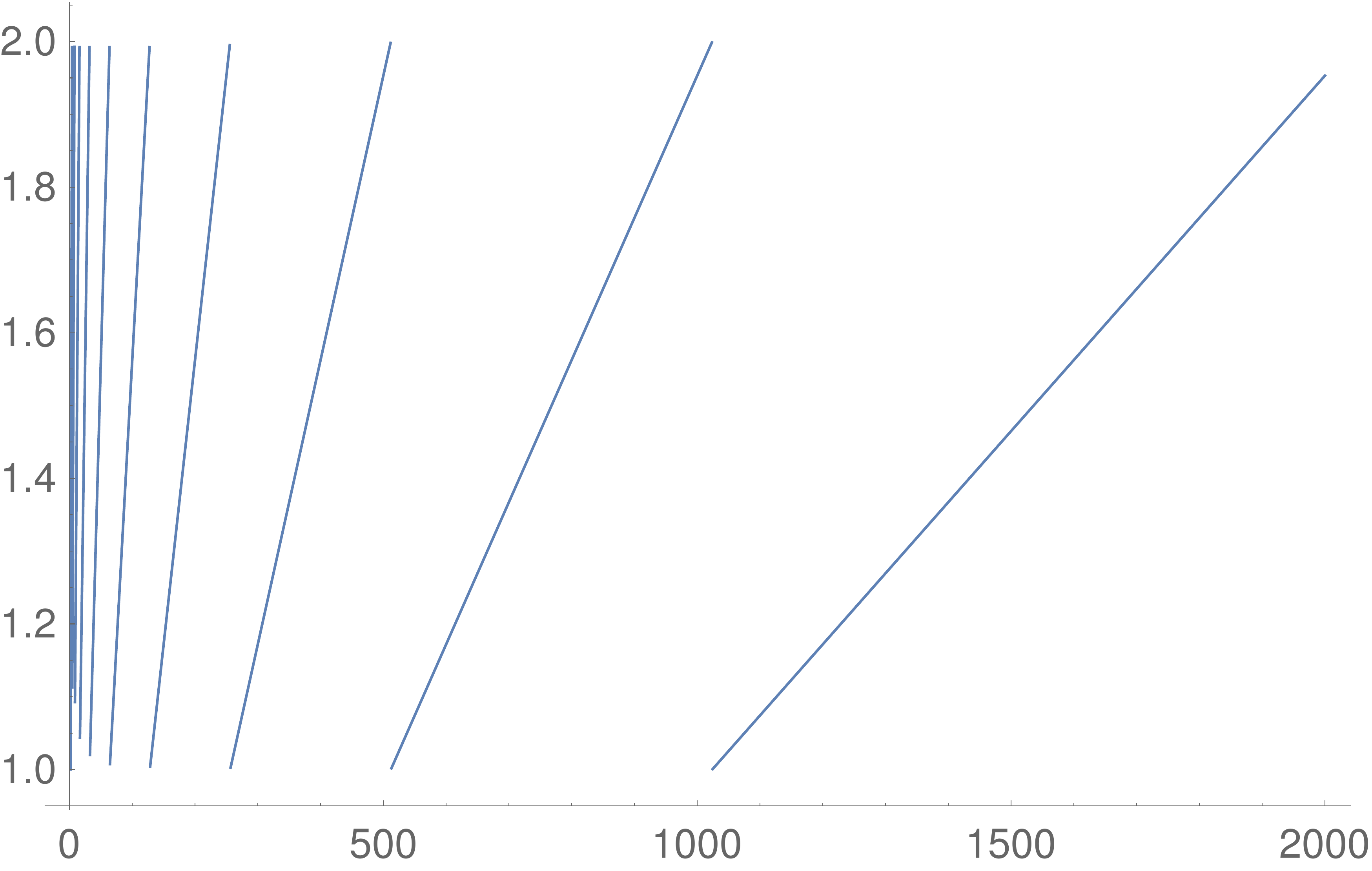}
\caption{A graph of enumeration $ \nu (n) = \frac{n}{2^{\lfloor \lg n \rfloor}} $ of the set $ D $.}
\label{fig:mapping}
\end{figure}

\medskip

$ D $ is a dense subset of the interval $ [1,2) $ of reals. Indeed, if $ x \in [1,2) $ then for every $ n \in \mathbb{N} $, $ \frac{\lfloor 2^n x \rfloor}{2^n} \in D $ and
\begin{equation} \label{eq:convD}
\lim _{n \rightarrow \infty} \frac{\lfloor 2^n x \rfloor}{2^n} = x.
\end{equation}

Hence, for any $ x \in [1,2) $, putting 
\begin{equation} \label{eq:n=floor}
 n = \lfloor 2^i x \rfloor   , 
\end{equation}
so that
\[ {\lfloor \lg n \rfloor}  =   {\lfloor \lg \lfloor 2^i x \rfloor \rfloor}  =
  {\lfloor \lg  2^i x  \rfloor}  =  { \lfloor i+ \lg   x  \rfloor}  ={i+ \lfloor  \lg   x  \rfloor}  = i \]
  [the last equality holds because $ 1 \leq x < 2 $ so that $ 0 \leq \lg   x < 1 $ and $ \lfloor  \lg   x  \rfloor = 0 $],
or
\begin{equation} \label{eq:lg=i}
{\lfloor \lg n \rfloor} = i ,
\end{equation}
we conclude, by virtue of \eqref{eq:convD},
\begin{equation} \nonumber 
\tilde{F}(x) = \tilde{F}(\lim _{i \rightarrow \infty} \frac{\lfloor 2^i x \rfloor}{2^i}) =
\end{equation}
[by the continuity of $ \tilde{F}(x) $]
\[ = \lim _{i \rightarrow \infty} \tilde{F}( \frac{\lfloor 2^i x \rfloor}{2^i}) = \]
[by the equality \eqref{eq:proof5} of Lemma~\ref{lem:cff_F_tilde}]
\[ = \lim _{i \rightarrow \infty} \frac{ \lfloor 2^i x \rfloor (i + 2) -  2 B(\lfloor 2^i x \rfloor)}{2^{i}} - 2 = 
 \lim _{i \rightarrow \infty} \frac{i \lfloor 2^i x \rfloor  -  2 B(\lfloor 2^i x \rfloor)}{2^{i}} +  \lim _{i \rightarrow \infty} 2\frac{ \lfloor 2^i x \rfloor  }{2^{i}} - 2 =\]
 [by the equality \eqref{eq:convD}]
 \[ \lim _{i \rightarrow \infty} \frac{i \lfloor 2^i x \rfloor  -  2 B(\lfloor 2^i x \rfloor)}{2^{i}} +  2x - 2.  \]
 Thus, for any $ x \in [1,2) $,
 \begin{equation} \label{eq:proof6}
\tilde{F}(x) =  \lim _{i \rightarrow \infty} \frac{i \lfloor 2^i x \rfloor  -  2 B(\lfloor 2^i x \rfloor)}{2^{i}}  +  2x - 2  .
 \end{equation}
The equality \eqref{eq:proof6} shows that if there is a $ cf\!f $ $ \beta : \mathbb{N} \longrightarrow  \mathbb{N} $ for function $ B $ defined by the equations \eqref{eq:rec_rel1} and \eqref{eq:rec_rel3} page~\pageref{eq:rec_rel1} then there is a $ cf\!f $ $ \varphi : [1,2) \longrightarrow [0,\frac{2}{3}] $  given by 
\[  \lim _{i \rightarrow \infty} \frac{i \lfloor 2^i x \rfloor  -  2 \beta(\lfloor 2^i x \rfloor)}{2^{i}}   +  2x - 2  \]
for   $  \tilde{F}(x)  $.
 This completes the proof of Lemma~\ref{lem:main_no_cff}.
\end{proof}

\medskip

Should the nowhere-differentiable function $  \tilde{F}  $ have a $ cf\!f $, it would be differentiable everywhere except, perhaps, on a non-dense subset of $ \mathbb{R} $. The following inductive argument demonstrates that. All atomic $ cf\!f $s are differentiable except, perhaps, on a non-dense subset of $ \mathbb{R} $. If a finite number of $ cf\!f $s are differentiable except, perhaps, on non-dense subsets of $ \mathbb{R} $ then their composition is differentiable except, perhaps, on non-dense subsets of $ \mathbb{R} $. \footnote{For instance, function $ \sqrt{x}Zigzag(\frac{1}{x}) $ is differentiable on $ [0,1] $, except for the non-dense set $ \{0\} \cup \{ \frac{1}{n} \mid n \in \mathbb{N} \} $.} Thus $  \tilde{F}  $ has no $ cf\!f $.

\medskip

The above observation, together with Lemma~\ref{lem:main_no_cff}, complete the proof of Theorem~\ref{thm:no_closed_form_add}.


\section{Proof of Theorem~\ref{thm:mainB} page~\pageref{thm:mainB}, Section~\ref{sec:compfromfrac} \label{sec:comment}}

In this Section, I provide an analytic proof of experimentally-derived Theorem~\ref{thm:mainB} page~\pageref{thm:mainB}, Section~\ref{sec:compfromfrac}. This result, re-stated by Theorem~\ref{thm:mainB_add} below, allows for practically efficient computations of values of the continuous Blackmange function  for reals with finite binary floating-point representations. I also provide some properties (Lammas~\ref{lem:100}, \ref{lem:200}, and \ref{lem:300}) of the \textit{Zigzag} function, given by the equality \eqref{eq:def_zigzag} page~\pageref{eq:def_zigzag} and visualized on Figure~\ref{fig:zigzag} page~\pageref{fig:zigzag}, that are useful for a neat derivation of a formula for the Blancmange function as (the limit of) a finite sum of some values of the \textit{Zigzag} function. 

\medskip

Let the function\footnote{Known as the Blancmange function.} 
$\tilde{F}  $, visualized on Figure~\ref{fig:Tak} page~\pageref{fig:Tak}, be defined by
\eqref{eq:def_lim2} page~\pageref{eq:def_lim2}, and $ B(n) $, given by \eqref{eq:rec_rel1} and \eqref{eq:rec_rel3} page~\pageref{eq:rec_rel1}, be the least number of comparisons of keys that $ {\tt MergeSort} $ performs while sorting an $ n $-element array.

\begin{thm} \label{thm:mainB_add} {\rm (Same as Theorem~\ref{thm:mainB}.)}
For every positive integer $ n $ \footnote{Of course, one if free to assume that $ n $ is odd here.} and integer $ k $ with $ n \leq 2^k $,
\begin{equation} \label{eq:mainB_add}
\tilde{F} (\frac{n}{2^k}) = \frac{n \times k - 2 B(n)}{2^k} .
\end{equation} 
\end{thm}

 \medskip 
 
 The reminder of this Section constitutes a proof of Theorem~\ref{thm:mainB_add}. 
 
 \medskip
 
 \textit{Note}. Function \textit{Zigzag}, visualized on Figure~\ref{fig:zigzag} page~\pageref{fig:zigzag}, has been defined by \eqref{eq:def_zigzag} page~\pageref{eq:def_zigzag}.

 \begin{lem} \label{lem:100}
 For every $ k \geq \lfloor \lg n \rfloor + 2 $,
 \begin{equation} \label{eq:100x}
 2^k Zigzag (\frac{n}{2^k}) = n.
 \end{equation}
 \end{lem}
 
 \begin{proof}
 Let $ k \geq \lfloor \lg n \rfloor + 2 $, or $ 2^k \geq 2 \times 2^{\lfloor \lg n \rfloor+1}  > 2n$, that is
 \begin{equation} \label{eq:110x}
 2^k > 2n.
 \end{equation}
  We have:
\[0 \leq \lfloor \frac{n}{2^k} \rfloor \leq \]  
[by \eqref{eq:110x}] 
\[\leq \lfloor \frac{n}{2n} \rfloor = \lfloor \frac{1}{2} \rfloor =0 \]
 or
\begin{equation} \label{eq:120}
 \lfloor \frac{n}{2^k} \rfloor = 0.
 \end{equation}
 Also,
\[1 \leq \lceil \frac{n}{2^k} \rceil \leq \]  
[by \eqref{eq:110x}] 
\[\leq \lceil \frac{n}{2n} \rceil = \lceil \frac{1}{2} \rceil =1 \]
 or
\begin{equation} \label{eq:130}
 \lceil \frac{n}{2^k} \rceil = 1.
 \end{equation}
 Now,
 \[2^k Zigzag (\frac{n}{2^k}) =\]
 [by \eqref{eq:def_zigzag} page~\pageref{eq:def_zigzag}]
 \[ = 2^k \min \{ \frac{n}{2^k} - \lfloor \frac{n}{2^k} \rfloor,\lceil \frac{n}{2^k} \rceil - \frac{n}{2^k} \} = \min \{ n - 2^k \lfloor \frac{n}{2^k} \rfloor,2^k \lceil \frac{n}{2^k} \rceil - n \} = \]
 [by \eqref{eq:120} and \eqref{eq:130}]
 \[ = \min \{ n ,2^k  - n \} = \]
 [since by  \eqref{eq:110x}, $ 2^k  - n \geq n $]
 \[=n.\]
 Hence, \eqref{eq:100x} holds.
 \end{proof}

\begin{lem} \label{lem:200}
For every $ k \geq \lfloor \lg n \rfloor + 1 $,
\begin{equation} \label{eq:200}
\sum_{i=\lfloor \lg n \rfloor + 2}^{k} 2^i Zigzag (\frac{n}{2^i}) = n \times k - n(\lfloor \lg n \rfloor + 1).
\end{equation}
\end{lem}
\begin{proof}
By induction on $ k $.

\medskip

\textit{Basis step}: $ k = \lfloor \lg n \rfloor + 1 $.

\medskip

\[ L = \sum_{i=\lfloor \lg n \rfloor + 2}^{\lfloor \lg n \rfloor + 1} 2^i Zigzag (\frac{n}{2^i}) = 0. \]
\[R = n (\lfloor \lg n \rfloor + 1) - n (\lfloor \lg n \rfloor + 1) = 0. \]
Hence, $ L=R $. This completes the \textit{Basis step}.

\bigskip

\textit{Inductive step}: $ k \geq \lfloor \lg n \rfloor + 2 $.

\medskip

\textit{Inductive hypothesis}: \eqref{eq:200}.

\medskip

\[ \sum_{i=\lfloor \lg n \rfloor + 2}^{k+1} 2^i Zigzag (\frac{n}{2^i}) = \sum_{i=\lfloor \lg n \rfloor + 2}^{k} 2^i Zigzag (\frac{n}{2^i}) +  2^k Zigzag (\frac{n}{2^k}) = \]
[by the \textit{Inductive hypothesis} and by the equality~\eqref{eq:100x} of Lemma~\ref{lem:100}]
\[= n \times k - n(\lfloor \lg n \rfloor + 1) +  n = n \times (k+1) - n(\lfloor \lg n \rfloor + 1) .\]
Thus
\[\sum_{i=\lfloor \lg n \rfloor + 2}^{k+1} 2^i Zigzag (\frac{n}{2^i}) = n \times (k+1) - n(\lfloor \lg n \rfloor + 1) .\]
This completes the \textit{Inductive step}.
\end{proof}

\begin{lem}  \label{lem:300}
For every $ k \geq \lfloor \lg n \rfloor + 1 $,
\begin{equation} \label{eq:300}
2^k \tilde{F} (\frac{n}{2^k}) =  \sum_{i=1}^{k} 2^i Zigzag (\frac{n}{2^i}).
\end{equation}
\end{lem}
\begin{proof}
By the definition  
\eqref{eq:def_lim2} page~\pageref{eq:def_lim2} of  function $\tilde{F}  $ , we get:
\[ 2^k \tilde{F} (\frac{n}{2^k}) = 2^k \sum _{i=0} ^{\infty} \frac{1}{2^i} Zigzag (2^i \frac{n}{2^k}) = \]
[since for every integer $ x $, $ Zigzag (x) = 0 $, so that for $ i \geq k $, $ Zigzag (2^i \frac{n}{2^k}) = 0 $]
\[ = 2^k \sum _{i=0} ^{k-1} \frac{1}{2^i} Zigzag (2^i \frac{n}{2^k}) = 
 \sum _{i=0} ^{k-1} 2^{k-i} Zigzag ( \frac{n}{2^{k-i}}) =
\]
[putting $ j = k - i $]
\[ = \sum _{j=1} ^{k} 2^{j} Zigzag ( \frac{n}{2^{j}}),\]
which completes the proof of \eqref{eq:300}.
\end{proof}

At this point, we are ready to conclude the proof of Theorem~\ref{thm:mainB}. 

\medskip

By virtue of \eqref{eq:formula_Zigzag_B} page~\pageref{eq:formula_Zigzag_B}, we have:
\[ 2 B(n) = n(\lfloor \lg n \rfloor + 1) -  \sum _{k=0} ^{\lfloor \lg n \rfloor} 2^{k+1}  Zigzag (\frac{n}{2^{k+1}}) = \]
\[ = n(\lfloor \lg n \rfloor + 1) - \sum _{i=1} ^{\lfloor \lg n \rfloor+1} 2^{i}  Zigzag (\frac{n}{2^{i}}) = \] 
\[ = n(\lfloor \lg n \rfloor + 1) - \sum _{i=1} ^{k} 2^{i}  Zigzag (\frac{n}{2^{i}})
+ \sum _{i=\lfloor \lg n \rfloor+2} ^{k} 2^{i}  Zigzag (\frac{n}{2^{i}}) = \] 
[by Lemmas \ref{lem:200} and \ref{lem:300}]
\[ = n(\lfloor \lg n \rfloor + 1) - 2^k \tilde{F} (\frac{n}{2^k})
+  n \times k - n(\lfloor \lg n \rfloor + 1) =  n \times k - 2^k \tilde{F} (\frac{n}{2^k}) ,\]
that is,
\[ 2 B(n) =  n \times k - 2^k \tilde{F} (\frac{n}{2^k}) ,\]
from which \eqref{eq:100x} follows.

\medskip

This completes the proof of Theorem~\ref{thm:mainB}.


\bigskip \bigskip

\textit{Note}. A glance at the proof of Lemma~\ref{lem:200} suffices to notice that it fails if $ n > 2^k $, and so does Theorem~\ref{thm:mainB}. In particular, for $ k = \lfloor \lg n \rfloor $, Lemma~\ref{lem:cff_F_tilde} page~\pageref{lem:cff_F_tilde} yields
\begin{equation}
\tilde{F} (\frac{n}{2^k}) = \frac{n \times k - 2 B(n)}{2^k} +  \frac{2n}{2^k} -2
> \frac{n \times k - 2 B(n)}{2^k} 
\end{equation}
since for $ n > 2^{\lfloor \lg n \rfloor} $, $ \frac{2n}{2^{\lfloor \lg n \rfloor}} -2 >0. $

\bibliographystyle{alpha}
\bibliography{ref}

\section*{APPENDIX} \label{appendix} 

\appendix

\section{\label{sec:mergesort}  {A derivation of the worst-case running time $ W(n) =  \sum_{i=1} ^{n} \lceil \lg i \rceil $ of} $ {\tt MergeSort} $} \label{sec:excerpts}

Let's assume that $ n \geq 2 $ is large enough to spur a cascade of many recursive calls to $ {\tt MergeSort} $ following the \textit{recursion tree} $ T $, a sketch of which 
is shown on Figure~\ref{fig:rectre}. 

\medskip

The nodes in tree $ T $ correspond to calls to $ {\tt MergeSort} $ and show sizes of (sub)arrays passed to those calls. The root corresponds to the original call to $ {\tt MergeSort} $. If a call that is represented by a node $ p $ executes further recursive calls to $ {\tt MergeSort} $ then these calls are represented by the children of $ p $; otherwise $ p $ is a leaf. Thus, $ T $ is a $ 2 $-tree\footnote{A binary tree whose every non-leaf has exactly $ 2 $ children.}.

\medskip

The levels in tree $ T $ are enumerated from $ 0 $ to $ h $, where $ h $ is the number of the last level of the tree, or - in other words - the depth of $ T $. On Figure~\ref{fig:rectre}, they are shown on the left side of the tree. The root is at the level $ 0 $, its children are at level $ 1 $, its grand children are at level $ 2 $, its great grand children (not shown on the sketch) are at level $ 3 $, at so on. 
Clearly, since every call to $ {\tt MergeSort} $ on a sub-array of size $ \geq 2$ executes two further recursive calls to $ {\tt MergeSort} $, only the nodes that show value $ 1 $ are leaves and all other nodes have $ 2 $ children each. Thus, since all nodes in the last level $ h $ are leaves, they all show value $ 1 $. And since the original input array gets split, eventually, onto $ n $ $ 1 $-element sub-arrays, the number of all leaves in $ T $ is $ n $. (This, however, does not mean that the last level $ h $ necessarily contains all the leaves of $ T $.)

\medskip



If a level $ i $ has $ 2^i $ nodes, each of them showing a value $ \geq 2 $, then each such node has $ 2 $ children so that level $ i+1 $ has twice the number of nodes in level $ i $, that is, $ 2^{i+1} $ nodes. Since level $ 0 $ has $ 2^0 $ nodes, it follows (completion of a proof by induction with the basis and inductive steps outlined above is left as an exercise for the reader) that if $ k $ is the level number of any level above which all the nodes show values $ \geq 2 $ then all levels $ i = 0,...k $ contain exactly $ 2^i $ nodes each. 

\medskip

The last level $ h $ may contain  $ 2^h $ nodes or less. We are going to show that each level $ i $ above level $ h $ contains exactly $ 2^i $ nodes. Here is a very insightful property that we are going to use for that purpose. It states that $ {\tt MergeSort} $ is splitting its input array fairly evenly so that at any level of the recursive tree, the difference between the lengths of the longest sub-array and the shortest sub-array is $ \leq 1. $ This fact is the root cause of 
good worst-case performance
of $ {\tt MergeSort} $.

\begin{property1} \label{pro:level} 
The difference between values shown by any two nodes in the same level of the recursion tree for $ {\tt MergeSort} $ is $ \leq 1 $.
\end{property1} 
\begin{proof} The Property clearly holds for level $ 0 $. We will show that if it holds for level $ i $ and $ i $ is not the last level of the recursion tree (that is, $ i < h $) then it also holds for the level $ i+1 $.

\smallskip

Let us assume that the Property holds for some level $ i < h $.
Let $ c \leq d $ be numbers shown by any two (not necessarily distinct) nodes in level $ i+1 $. It suffices to show that 
\begin{equation} \label{eq:prop0} 
d-c\leq 1.
\end{equation}
Let $ a \leq b $ be the numbers shown by the parents of the mentioned above nodes. Those parents, of course, must reside in the level $ i $. By the inductive hypothesis (that holds for level $ i $), $ b - a \leq 1 $, that is,
\begin{equation} \label{eq:prop1} 
a \leq b \leq a+1.
\end{equation} 
The numbers shown by all their four children are  $ \lfloor \frac{a}{2} \rfloor $, $ \lceil \frac{a}{2} \rceil $, $ \lfloor \frac{b}{2} \rfloor $ and $ \lceil \frac{b}{2} \rceil $, respectively, so the largest difference between any of those four numbers is $ \lceil \frac{b}{2} \rceil - \lfloor \frac{a}{2} \rfloor $. In particular, $ d - c $ is not larger than that. 
We have:
 \[ d - c \leq \lceil \frac{b}{2} \rceil - \lfloor \frac{a}{2} \rfloor \leq \]
 [by (\ref{eq:prop1})]
 \[ \leq \lceil \frac{a+1}{2} \rceil - \lfloor \frac{a}{2} \rfloor  \]
 [since for any integer $ c $, $  \lceil \frac{c}{2} \rceil = \lfloor \frac{c+1}{2} \rfloor $]
 \[ = \lfloor \frac{a+2}{2} \rfloor  - \lfloor \frac{a}{2} \rfloor =
 \lfloor \frac{a}{2} + 1 \rfloor  - \lfloor \frac{a}{2} \rfloor = \]
 [since for every $ x $, $ \lfloor x + 1 \rfloor = \lfloor x \rfloor + 1 $]
\[ =  \lfloor \frac{a}{2} \rfloor +1 - \lfloor \frac{a}{2} \rfloor = 1. \] 
Thus (\ref{eq:prop0}) holds.
 This completes the inductive step and completes the proof of the Property.
\end{proof}

\medskip

As we have noted, the values shown at all nodes in the last level $ h$  are all $ 1 $. Thus the values shown at their parents, that reside at level $ h-1 $ are all $ 2 $, and the values shown at their grand parents, that reside at level $ h-2 $ are all $ \geq 3 $. Thus, by Property~\ref{pro:level}, all nodes at level $ h-2 $ show values $ \geq 2 $, and, therefore (as we have proved before), all levels $ i = 0, ... , h-1 $ have $ 2^i $ nodes, each, as it has been visualized on Figure~\ref{fig:rectre}.

\begin{theorem1} \label{thm:depthRectree} 
The depth $ h $ of the recursion tree $ T (n) $ for 
$ {\tt MergeSort} $ run on an array of size $ n $ is
\begin{equation} \label{eq:depthRectree}
h = \lceil \lg n \rceil.
\end{equation}
\end{theorem1}
\begin{proof}
Since every level of $ T $, except, perhaps, for the last level, has the maximal number of nodes, a 2-tree  with $ n $ leaves could not be any shorter than  $ T $. So, $ T $ is a shortest 2-tree with $ n $ leaves. Therefore (by a well known fact), its depth $ h $ is equal to $ \lceil \lg n \rceil $. Thus (\ref{eq:depthRectree}) holds.
\end{proof}

Because each node in any level above $ h-1 $ shows value $ \geq 2 $, it has $ 2 $ children. Thus the value it shows is equal to the sum of values shown by its children, as we have indicated at the beginning of this section. From that we conclude (a proof by induction is left as an exercise for the reader) that the sum of values shown at nodes in any level $ i = 0, ..., h-1 $ is the same for each such level. Thus the said sum is equal to the value showed by the only node at level 0, that is, is equal to $ n $.

\medskip

Let $ a_1,...,a_{2^i} $ be the values shown at the nodes of some level $ i = 0, ..., h-1 $. The number of comps performed by a call to $ {\tt Merge }$ invoked by the call to $ {\tt MergeSort }$ on an array of $ a_j $ elements is either $ 0 $ if $ a_j = 1 $ (no call to $ {\tt Merge }$ is made) or, as we have shown in the previous section, is $ a_j -1 $ if $ a_j \geq 2 $. So, in either case, it is $ a_j -1 $.
Thus the number $  $ of comps $ C_i $ performed at  level $ i  $ is 
\begin{equation} \label{eq:levComp} 
C_i = (a_1 - 1) + ... + (a_{2^i} - 1) = (a_1 + ... + a_{2^i}) - \underbrace{(1+ ... + 1)}_{2^i} = n - 2^i.
\end{equation}
Moreover, since all nodes at the last level $ h $ are 1's
\begin{equation} \label{eq:levComp9} 
C_h = 0.
\end{equation}
Therefore, the total number $ W(n) $ of comps that $ {\tt MergeSort }$ performs in the worst case on an $ n $-element array is equal to
\[W(n) = \sum _{i=0}^{h} C_i =   \]
[by (\ref{eq:levComp9})]
\[ = \sum _{i=0}^{h-1} C_i =  \]
[by (\ref{eq:levComp})]
\[ = \sum _{i=0}^{h-1} ( n - 2^i) = n h - (2^h - 1) = n h - 2^h + 1 = \]
[by (\ref{eq:depthRectree})]
\[ =n \lceil \lg n \rceil - 2^{\lceil \lg n \rceil} + 1 .\]
This way I have proved the following.
\begin{theorem} \label{thm:w-cMS}
The number $ W(n) $ of comparisons of keys that $ {\tt MergeSort} $ performs in the worst case while sorting an $ n $-element array is
\begin{equation} \label{eq:levComp2} 
W(n) = n \lceil \lg n \rceil - 2^{\lceil \lg n \rceil} + 1 .
\end{equation}
\end{theorem}
\textbf{Proof} follows from the above derivation.
\hfill $ \Box $

\medskip

Using the well-known\footnote{See \cite{knu:art}.} closed-form formula for $ \sum_{i=1} ^{n} \lceil \lg i \rceil $,
 I conclude that
\begin{equation} \label{eq:recmergesort900} 
W(n) =  \sum_{i=1} ^{n} \lceil \lg i \rceil.
\end{equation}

\section{Proof of $ \sum _{i=0} ^{m-1} \lfloor \frac{n+i}{m} \rfloor = n $} \label{sec:proof1}

\begin{thm1}
For every natural number n and every positive natural number m,
\[ \sum _{i=0} ^{m-1} \lfloor \frac{n+i}{m} \rfloor = n. \]
\end{thm1}

\begin{proof} Let $n = k m + l$, where $0 \leq l < m$.

\medskip

We have

\[\lfloor \frac{n+i}{m} \rfloor = \lfloor \frac{k m + l +i}{m} \rfloor = \lfloor k + \frac{l +i}{m} \rfloor
= k + \lfloor \frac{l +i}{m} \rfloor . \]

\medskip

Therefore,

\[  \sum _{i=0} ^{m-1} \lfloor \frac{n+i}{m} \rfloor = mk + \sum _{i=0} ^{m-1} \lfloor \frac{l +i}{m} \rfloor =
mk + \sum _{i=m-l} ^{m-1} \lfloor \frac{l +i}{m} \rfloor = mk + \sum _{i=m-l} ^{m-1} 1 = mk + l = n. \]
\end{proof}

\bibliographystyle{siam}
\bibliography{ref}
\bigskip
\bigskip
\copyright \textit{2016 Marek A. Suchenek. All rights reserved by the author. \newline A non-exclusive license to distribute this article is granted to arXiv.org}.

\end{document}